\newif\ifarxiv

\arxivtrue

\ifarxiv
\documentclass{article}
\usepackage[numbers]{natbib}
\usepackage{amsfonts,
	amssymb,
	xcolor,
	booktabs,
	caption,
	amsthm,
	amsmath,
	url}
\usepackage[margin=1.25in]{geometry}
\else
\documentclass[sigconf]{acmart}
\fi

\usepackage{enumitem,setspace}
\usepackage[textsize=tiny,disable=true]{todonotes}

\usepackage{algorithm}
\usepackage[font=small,skip=1pt]{caption}

\newcommand{\widesim}[2][1.5]{
  \mathrel{\overset{#2}{\scalebox{#1}[1]{$\sim$}}}
}

\usetikzlibrary{
  positioning,
  automata,
  arrows,
  calc,
  decorations.pathmorphing,
  decorations.markings,
  decorations.pathreplacing,
  patterns,
  shapes.arrows,
  shapes.callouts,
  topaths
}

\tikzset{
    linpt/.style={fill,rectangle,yscale=0.085cm,xscale=0.012cm},
    >=stealth',
    pil/.style={
           ->,
           thick,
           shorten >=4pt
           }
}

\newtheorem{theorem}{Theorem}
\newtheorem{lemma}[theorem]{Lemma}
\newtheorem{observation}[theorem]{Observation}

\AtBeginDocument{%
  \providecommand\BibTeX{{%
    \normalfont B\kern-0.5em{\scshape i\kern-0.25em b}\kern-0.8em\TeX}}}

\ifarxiv
\else
\copyrightyear{2024} 
\acmYear{2024} 
\setcopyright{acmlicensed}\acmConference[PODC '24]{ACM Symposium on Principles of Distributed Computing}{June 17--21, 2024}{Nantes, France}
\acmBooktitle{ACM Symposium on Principles of Distributed Computing (PODC '24), June 17--21, 2024, Nantes, France}
\acmDOI{10.1145/3662158.3662775}
\acmISBN{979-8-4007-0668-4/24/06}

\fi
\begin{document}

\title{Determining Recoverable Consensus Numbers}

\ifarxiv
\author{Sean Ovens \\
	University of Waterloo\\
	\texttt{s2ovens@uwaterloo.ca} \\
}
\fi


\ifarxiv
\maketitle
\else
\fi

\begin{abstract}
Herlihy's wait-free consensus hierarchy classifies the power of object types in asynchronous shared memory systems where processes can permanently crash (i.e. stop taking steps).
In this hierarchy, a type has consensus number $n$ if objects of that type can be used along with (read/write) registers to solve consensus among $n$ processes that can permanently crash, but not among $n+1$ or more processes.
In systems where processes can recover after crashing, the power of an object type to solve consensus may be different.
Golab's recoverable consensus hierarchy classifies the power of object types in such a system.
In the recoverable consensus hierarchy, a type has recoverable consensus number $n$ if objects of that type can be used along with registers to solve consensus among $n$ processes that can recover after crashing, but not among $n+1$ or more processes.
In this paper, we prove that the recoverable consensus hierarchy of deterministic, readable types is robust, i.e., if consensus can be solved among $n$ processes that can recover after crashing using a collection of objects of deterministic, readable types, then one of these types has recoverable consensus number at least $n$.
This is important for comparing the relative computational power of different deterministic, readable types, because it implies that one cannot combine various objects to obtain an algorithm that is better at solving recoverable consensus than any of the individual object types.
Our result can be used to show that, for all $n \geq 4$, there exists a readable type with consensus number $n$ and recoverable consensus number $n-2$.
We also show that, for all $n > n' \geq 1$, there exists a non-readable type that has consensus number $n$ and recoverable consensus number $n'$.
\end{abstract}

\section{Introduction}\label{sec:intro}

Consensus has been crucial in developing our understanding of what can be implemented using various combinations of shared object types.
In the consensus problem, processes begin with private inputs and attempt to collectively agree on a single output, which must be equal to one of the inputs.
A solution to consensus is wait-free if every process terminates within a finite number of its own steps.
The wait-free consensus hierarchy \cite{h-91} can be used to compare the relative computational power of various types in asynchronous shared memory models.
It also allows us to reason about what can be implemented in a wait-free manner using certain types; if $\mathcal{T}_1$ has consensus number $n$ and $\mathcal{T}_2$ has consensus number $m < n$, then it is impossible to use objects of $\mathcal{T}_2$ along with registers to implement an object of $\mathcal{T}_1$ in a wait-free manner for more than $m$ processes.
Furthermore, Herlihy \cite{h-91} showed that $n$-process consensus is universal for systems with $n$ processes in the following sense: every object has an $n$-process wait-free implementation using objects with consensus number at least $n$ along with registers.

In this paper, we study Golab's~\cite{g-20} recoverable consensus hierarchy, which classifies the power of objects to solve consensus among processes that can recover after crashing.
When a process crashes, its local variables (including its program counter) are all reset to their initial values.
However, all shared objects retain their values in the wake of a crash.
This is a standard asynchronous shared memory model for systems containing non-volatile main memory \cite{g-20,abh-18,gr-16,dffr-22}, which retains its state when it loses power.
A recoverable consensus algorithm satisfies recoverable wait-freedom if a process that executes its algorithm starting from its initial state either crashes or outputs a value after a finite number of its own steps.
Therefore, unlike a wait-free consensus algorithm, infinitely long executions are admissible as long as each process that takes infinitely many steps also crashes infinitely often.
Notice that a recoverable consensus algorithm solves wait-free consensus if process crashes are permanent.
Therefore, the consensus number of any type is greater than or equal to its recoverable consensus number.

There are two main kinds of crashes that have been considered.
With simultaneous crashes, all of the processes always crash at the same time.
Simultaneous crashes could be used to model power failures in multiprocessing systems, for example.
It is known that the consensus number of an object is the same as its recoverable consensus number if all crashes are simultaneous \cite{g-20,dffr-22}.
Therefore, in this paper we focus on individual crashes, where any process can crash at any time.
This may be more realistic for modelling crashes in distributed systems where processes may reside on physically distinct machines.
Golab~\cite{g-20} showed that the recoverable wait-free consensus hierarchy with individual crashes is different from Herlihy's \cite{hs-99} wait-free consensus hierarchy by proving that test-and-set objects, which have consensus number $2$, cannot be used along with registers to solve recoverable wait-free consensus between $2$ processes if individual crashes are allowed.

Like wait-free consensus, recoverable consensus is known to be universal; Berryhill, Golab, and Tripunitara~\cite{bgt-15} showed that Herlihy's \cite{h-91} universal construction also works in the recoverable setting with simultaneous crashes by placing part of the data structure in non-volatile memory.
Even when processes can crash individually, Delporte-Gallet, Fatourou, Fauconnier, and Ruppert \cite{dffr-22} showed that recoverable consensus is universal.
More specifically any object can be implemented in a recoverable wait-free manner using objects with recoverable consensus number at least $n$ along with registers.
Implementations obtained by this construction provide detectability \cite{fhmp-18}; roughly speaking, if an operation $op$ is interrupted by a crash, then the process $p_i$ that invoked $op$ can tell upon recovery whether $op$ linearized or not before the crash and, if it did, $p_i$ can determine its response.

\smallskip

Delporte-Gallet, Fatourou, Fauconnier, and Ruppert \cite{dffr-22} recently studied the relationship between the solvability of wait-free consensus and recoverable wait-free consensus.
In particular, they focused on object types that are \emph{deterministic} and \emph{readable}.
A type is \emph{deterministic} if, for any object and any operation of that type, when the operation is applied to the object, its response and the resulting value of the object is uniquely determined by the object's current value and the operation that was applied.
A type is \emph{readable} if it supports an operation that returns its current value and does not change its value.
The authors defined a characterization of deterministic, readable objects that can be used along with registers to solve recoverable wait-free consensus among $n$ processes:
Informally, a type is called $n$-recording if there exist a value $u$, $n$ operations $o_0, \ldots, o_{n-1}$ assigned to processes $p_0, \ldots, p_{n-1}$, respectively, and a partition of the processes into two teams such that, if some subset of the processes apply their operations to an object of that type with initial value $u$, then the resulting value of the object records the team of the first process to apply its operation.
They showed that objects of any $n$-recording, deterministic, readable type can be used, along with registers, to solve recoverable wait-free consensus among $n$ processes.
In other words, the recoverable consensus number of any $n$-recording, deterministic, readable type is at least $n$.
Furthermore, they showed that any deterministic, readable type $\mathcal{T}$ with consensus number $n \geq 4$ is $(n-2)$-recording.
Therefore, the recoverable consensus number of $\mathcal{T}$ is between $n-2$ and $n$.
Before our work, it was unknown whether there exists a type that has consensus number $n$ and recoverable consensus number $n-2$.

Delporte-Gallet, Fatourou, Fauconnier, and Ruppert \cite{dffr-22} also nearly proved that $n$-recording is necessary for solving recoverable wait-free consensus among $n$ processes.
More specifically, they showed that if recoverable wait-free consensus can be solved among $n$ processes using objects of types $\mathcal{T}_0, \mathcal{T}_1, \mathcal{T}_2, \ldots$, then at least one of these types $\mathcal{T}_i$ is $(n-1)$-recording.
Until our results, it remained open to determine whether at least one of the types $\mathcal{T}_i$ is $n$-recording.
For all $n \geq 4$, the authors also defined a deterministic, readable type $\mathcal{X}_n$ that has consensus number $n$, is $(n-2)$-recording, and is not $(n-1)$-recording.
Hence, the recoverable consensus number of $\mathcal{X}_n$ is at most $n$ and at least $n-2$.

Unlike recoverable wait-free consensus, there is a known condition that characterizes deterministic, readable types that can solve wait-free consensus among $n$ processes:
Ruppert \cite{r-00} showed that a deterministic, readable type has consensus number at least $n$ if and only if it is $n$-discerning.
It is possible to decide whether a type is $n$-discerning in a finite amount of time, as long as the type has finite numbers of values and operations.
Furthermore, Ruppert showed that if wait-free consensus can be solved using objects of types $\mathcal{T}_0, \mathcal{T}_1, \mathcal{T}_2, \ldots$, then at least one of these types $\mathcal{T}_i$ is $n$-discerning.
In other words, it is impossible to have a set of deterministic, readable types whose power to solve wait-free consensus exceeds the power of any individual type in the set.
This property of the wait-free consensus hierarchy of readable types is called \emph{robustness}.

Robustness is a desirable property because it allows us to determine the power of a system containing many different types by studying the power of the individual types.
Robust wait-free hierarchies were first studied by Jayanti \cite{j-97}, who showed that a restricted version of the wait-free consensus hierarchy is not robust; the restricted version of the hierarchy classifies types based on the number of processes among which a single object of that type (along with infinitely many registers) can solve wait-free consensus.
Schenk \cite{sthesis-96} and Lo and Hadzilacos \cite{lh-00} later showed that the unrestricted\todo{is this the right word?} wait-free consensus hierarchy is not robust if non-deterministic objects are allowed.
On the other hand, Herlihy and Ruppert \cite{hr-00} showed that the consensus hierarchy is robust for the set of all deterministic object types that can be accessed at most once by each process.


\smallskip

This paper presents the following contributions.
Our main result is that the recoverable consensus hierarchy is robust for the set of all deterministic, readable types.
To prove this, we show that if objects of deterministic types $\mathcal{T}_0, \mathcal{T}_1, \mathcal{T}_2, \ldots$ can be used along with registers to solve recoverable wait-free consensus among $n$ processes, then at least one of these types $\mathcal{T}_i$ is $n$-recording.
By the result of \cite{dffr-22} mentioned previously, if $\mathcal{T}_i$ is deterministic and readable, then objects of $\mathcal{T}_i$ along with registers can be used to solve recoverable wait-free consensus among $n$ processes.
Furthermore, our result implies that the type $\mathcal{X}_n$ has recoverable consensus number $n-2$.
In our proof, we adapt the idea of a valency argument to the recoverable setting.
Previous work on recoverable consensus has attempted to do the same by considering only executions in which a fixed set of processes do not crash, and then constructing such an execution in which at least one of these processes takes an infinite number of steps without outputting a value.
Our proof generalizes this approach by only considering executions in which the maximum number of times each process  $p_i$ can crash is proportional to the number of steps taken by processes $p_0, \ldots, p_{i-1}$ with lower identifiers.
More specifically, the number of times process $p_i$ crashes is $f(n)$ times the number of steps taken by $p_0, \ldots, p_{i-1}$, for some function $f$.
In particular, $p_0$ never crashes.
Intuitively, processes with smaller identifiers have higher priority in the sense that they are allowed to crash less often than processes with larger identifiers.
This way, we avoid fixing a set of processes that do not crash, but in any infinitely long execution, some process takes infinitely many steps without crashing.
We believe that this approach could be used to obtain other impossibility and complexity results for recoverable algorithms.

Next, we show that the power to solve consensus of objects that do not support a read operation can decrease arbitrarily when crash recovery is possible.
More specifically, we show that, for all $n > n' \geq 1$, there exists a deterministic type $\mathcal{T}_{n, n'}$ with consensus number $n$ and recoverable consensus number $n'$.
It remains open to determine whether the recoverable consensus hierarchy is robust for the set of all deterministic objects.

\smallskip

We discuss our model of computation in Section~\ref{sec:model}.
Section~\ref{sec:nec} contains our proof that the recoverable consensus hierarchy of deterministic, readable objects is robust.
In Section~\ref{sec:nonreadable}, we define the type $\mathcal{T}_{n, n'}$ and prove that it has consensus number $n$ and recoverable consensus number $n'$.
Finally, we conclude with a brief discussion of future directions in Section~\ref{sec:conclusion}.

\section{Model}\label{sec:model}

We use a standard asynchronous shared memory model in which $n$ processes $p_0, \ldots, p_{n-1}$ communicate using shared \emph{objects}.
The number $i$ is called process $p_i$'s \emph{identifier}.
Each object has a \emph{type}, which defines a set of \emph{values}, a set of \emph{operations} that can be applied to an object of the type, and a set of \emph{responses} that these operations can return.
Every type has a \emph{sequential specification} that defines, for each value $v$ and each operation $op$ of the type, the response to that operation and a resulting value when $op$ is applied to an object of the type with value $v$.
In this paper, we consider only \emph{deterministic} types where, for any value $v$ and any operation $op$ of the type, there is exactly one possible response to $op$ and exactly one possible resulting value when $op$ is applied to an object of the type with value $v$.
A type is \emph{readable} if it supports the \emph{Read} operation; when applied to an object, the \emph{Read} operation returns the current value of the object without changing its value.

In a \emph{task}, each process begins with an input value from some domain and the processes produce output values such that the set of all output values satisfies some specification.
In the \emph{consensus} task, processes begin with input values from $\{0, 1\}$ and the processes that produce output values must satisfy the following properties:

\begin{itemize}
\item \emph{Agreement}: no two processes output different values.
\item \emph{Validity}: every output is the input of some process.
\end{itemize}


A \emph{configuration} of a consensus algorithm consists of a state for each process and a value for each object.
An \emph{output state} of a process denotes the value that process outputs when it enters the state.
An \emph{algorithm} defines a set of objects, an initial value for each of these objects, and an initial state for each process.
Furthermore, for every state of every process, an algorithm defines the next step that process will apply.
A \emph{step} can be an operation applied to some object or a no op.
When a process takes a step that applies an operation to an object in a configuration, its operation returns a response which depends on the operation that was applied, the value of the object in the configuration, and the sequential specification of the object's type.
Immediately after applying the operation, the process performs some local computation and possibly updates its state, resulting in a new configuration.
If a process takes a step when it is in an output state, that step is always a no op, which leaves the configuration unchanged.
In an \emph{initial configuration}, every process is in its initial state and every object has its initial value.
We use $\textit{value}(O, C)$ to denote the value of the object $O$ in the configuration $C$.
If $v \in \{0, 1\}$, then we use $\bar{v}$ to denote $1 - v$.

An \emph{execution} consists of an alternating sequence of configurations and \emph{events}, each of which is either a step or a crash of some process.
When a process crashes, it is reset to its initial state.
Every execution begins with a configuration, and every finite execution ends with a configuration.
Each event in an execution is applied in the configuration that precedes it, which results in the configuration that follows the step.
If $C$ is a configuration and $\alpha$ is a finite execution starting from $C$, then we use $C\alpha$ to denote the configuration at the end of $\alpha$.
If $C' = C\alpha$, then we say $C'$ is \emph{reachable from $C$ via $\alpha$}.
If $\alpha$ is a finite execution starting from $C$ and $\beta$ is an execution starting from $C\alpha$, then we use $\alpha\beta$ to denote the execution obtained by removing the first configuration from $\beta$ and appending the resulting sequence to $\alpha$.
An \emph{empty execution} contains no steps.
When a process $p_i$ takes a step leading to a configuration $C'$ in which $p_i$ has an output state corresponding to the value $v$, we say that $p_i$ outputs $v$ in that step.
Furthermore, for every execution $\alpha'$ starting from $C'$ (including the empty execution from $C'$), we say that $p_i$ has output the value $v$ in $C'\alpha'$.
In the context of consensus, we say that $p_i$ has \emph{decided} the value $v$ in $C'\alpha'$.
A \emph{solo-terminating} execution of some process $p_i$ contains only steps by $p_i$ and ends with a configuration in which $p_i$ has output a value.
An execution is produced by an \emph{adversary}, who decides which process will take the next step in each configuration.
The adversary also decides if and when processes crash.

A \emph{schedule} is a sequence of processes and crashes.
We use $c_i$ to denote a crash by process $p_i$.
The \emph{schedule of an execution} is the sequence of processes that take steps and crashes that occur in that execution.
If $\sigma$ is a schedule and $C$ is a configuration, then we use $\textit{exec}(C, \sigma)$ to denote the execution constructed by performing the steps and crashes in $\sigma$, starting with the configuration $C$.
We use $C\sigma$ to denote the final configuration in $\textit{exec}(C, \sigma)$.
For example, $Cp_ip_jc_i$ is the configuration that results when $p_i$ takes its next step in $C$, then $p_j$ takes its next step, and then $p_i$ crashes.
If $\alpha$ is an execution starting from $C$ and $\sigma$ is a schedule, then we use $\alpha\sigma$ to denote the execution $\alpha\;\textit{exec}(C\alpha, \sigma)$.

The configurations $C$ and $C'$ are \emph{indistinguishable} to a set of processes $\mathcal{Q}$, denoted by $C \widesim{\mathcal{Q}} C'$, if every process in $\mathcal{Q}$ has the same state in $C$ and $C'$.
Consider two executions $\alpha$ and $\alpha'$ from $C$ and $C'$, respectively.
The executions $\alpha$ and $\alpha'$ are \emph{indistinguishable} to $\mathcal{Q}$ if $C \widesim{\mathcal{Q}} C'$, $\alpha$ and $\alpha'$ contain only events by $\mathcal{Q}$, the objects accessed by $\alpha$ and $\alpha'$ have the same values in $C$ and $C'$, and the schedules of $\alpha$ and $\alpha'$ are the same.
If $C \widesim{\mathcal{Q}} C'$ and $\alpha$ is an execution from $C$ containing only events by $\mathcal{Q}$, and all of the objects accessed by $\mathcal{Q}$ have the same values in $C$ and $C'$, then there exists an execution $\alpha'$ from $C'$ containing only events by $\mathcal{Q}$ such that $\alpha \widesim{\mathcal{Q}} \alpha'$ \cite{ae-14}.

In a \emph{wait-free} consensus algorithm, every process outputs a value after taking a finite number of steps.
A type has \emph{consensus number $n$} if there is a wait-free consensus algorithm for $n$ processes using objects of that type along with registers, but there is no such algorithm for $n+1$ or more processes.
\emph{Recoverable wait-freedom} is a progress condition requiring that every process either crashes or outputs a value after taking a finite number of steps from its initial state.
A consensus algorithm that satisfies recoverable wait-freedom is called a \emph{recoverable wait-free consensus algorithm}.
A type has \emph{recoverable consensus number $n$} if there is a recoverable wait-free consensus algorithm for $n$ processes using objects of that type along with registers, but there is no such algorithm for $n+1$ or more processes.

\medskip

For all $\mathcal{P}' \subseteq \{p_0, \ldots, p_{n-1}\}$, define $\mathcal{S}(\mathcal{P}')$ as the set of schedules that contain at most one instance of every process in $\mathcal{P}'$.
For instance, $\mathcal{S}\bigl(\{p_0, p_2\}\bigr) = \bigl\{\langle\rangle, p_0, p_2, p_0p_2, p_2p_0\bigr\}$, where $\langle\rangle$ denotes the empty schedule.
The following definition is adapted from Ruppert~\cite{r-00}.
A deterministic type $\mathcal{T}$ is \emph{$n$-discerning} if there exist
\begin{itemize}
\item a value $u$ of $\mathcal{T}$,
\item a partition of the processes $\{p_0, \ldots, p_{n-1}\}$ into two nonempty sets $T_0$ and $T_1$, and
\item an operation $o_i$ of type $\mathcal{T}$ for each process $p_i$
\end{itemize}
\noindent such that, for all $j \in \{0, \ldots, n-1\}$, $R_{0, j} \cap R_{1, j} = \emptyset$, where $R_{x, j}$ is the set of pairs $(r, v)$ for which there is a schedule of processes $\sigma \in \mathcal{S}\bigl(\{p_0, \ldots, p_{n-1}\}\bigr)$ where $p_j \in \sigma$, the first process in $\sigma$ is in $T_x$, and, if the processes in $\sigma$ apply their operations in order on an object of $\mathcal{T}$ with initial value $u$, then $r$ is the response of $p_j$'s operation $o_j$ and $v$ is the resulting value of the object.
Ruppert proved that a deterministic, readable type has consensus number $n$ if and only if it is $n$-discerning.

The next definition is adapted from Delporte-Gallet, Fatourou, Fauconnier, and Ruppert \cite{dffr-22}.
A deterministic type $\mathcal{T}$ is \emph{$n$-recording} if there exist
\begin{itemize}
\item a value $u$ of $\mathcal{T}$,
\item a partition of the processes $\{p_0, \ldots, p_{n-1}\}$ into two nonempty sets $T_0$ and $T_1$, and
\item an operation $o_i$ of type $\mathcal{T}$ for each process $p_i$
\end{itemize}
\noindent such that the following two properties are satisfied:
\begin{itemize}
\item $U_0 \cap U_1 = \emptyset$, where $U_x$ is the set of values $v$ for which there is a schedule of processes $\sigma \in \mathcal{S}\bigl(\{p_0, \ldots, p_{n-1}\}\bigr)$ where the first process in $\sigma$ is in $T_x$ and, if the processes in $\sigma$ apply their operations in order on an object of $\mathcal{T}$ with initial value $u$, then $v$ is the resulting value of the object, and
\item if $u \in U_x$, then $|T_{\bar{x}}| = 1$.
\end{itemize}


\section{$n$-Recording is Necessary for Readable Types}\label{sec:nec}

In this section we prove that, for any $n \geq 2$, if there is an $n$-process recoverable wait-free consensus algorithm using objects with deterministic types $\mathcal{T}_0, \mathcal{T}_1, \mathcal{T}_2 \ldots$, then at least one of these types $\mathcal{T}_i$ is $n$-recording.
Note that this applies to any deterministic types, not just readable ones.
Theorem~8 of \cite{dffr-22} says that, if a deterministic, readable type $\mathcal{T}$ is $n$-recording, then the recoverable consensus number of $\mathcal{T}$ is at least $n$.
Combining this with our result proves that the recoverable wait-free consensus hierarchy is robust for the set of all deterministic, readable types.

\smallskip

We use a valency argument to obtain our result in this section.
Valency arguments were first used by Fischer, Lynch, and Paterson \cite{flp-85} to show that solving wait-free binary consensus among $2$ or more processes is impossible in the asynchronous message passing model, even if only $1$ process may crash.
They define a configuration as bivalent if, for all $v \in \{0, 1\}$, there exists an execution from that configuration in which $v$ is decided by some process.
If a configuration is not bivalent, then it is univalent.
A valency argument may also be used to show that solving wait-free binary consensus among $2$ or more processes is impossible in the asynchronous shared memory model using only registers as follows.
First, using validity, it can be argued that any initial configuration in which some process has input $0$ and another has input $1$ must be bivalent.
Let $C$ be a bivalent initial configuration.
Since the algorithm is wait-free, there must be an execution from $C$ that leads to a critical configuration $C'$, that is, a bivalent configuration such that, if any process takes a step in that configuration, then the resulting configuration is univalent.
Since $C'$ is bivalent, there must be a process $p_0$ such that every execution from $C'p_0$ decides $0$ and a process $p_1$ such that every execution from $C'p_1$ decides $1$.
Using a case analysis on which register $p_0$ and $p_1$ are poised to access in $C'$ and which operation they are poised to apply, it can be shown that $C'p_0$ and $C'p_1$ must have the same valency, which is a contradiction.

As noted by Golab \cite{g-20}, valency arguments are more difficult to apply when processes may recover after crashing because infinitely long executions in which no process decides are permissible as long as every process that takes infinitely many steps in the execution also crashes infinitely often.
Furthermore, if we construct a ``critical'' configuration $C'$ by considering only crash-free executions, then it might be the case that, for every process $p_i$, the configuration $C'p_i$ is $v$-univalent; in other words, any execution from $C'$ in which $\bar{v}$ is decided begins with a crash.
To get around this, previous authors \cite{g-20,dffr-22} have considered executions in which only a fixed set of processes may crash.
We generalize this approach by considering executions in which the number of crashes that a process may experience is proportional the number of steps taken by processes with smaller identifiers.

\medskip

Consider a recoverable wait-free consensus algorithm for the $n \geq 2$ processes $p_0, \ldots, p_{n-1}$ using objects of deterministic types $\mathcal{T}_0, \mathcal{T}_{1}, \mathcal{T}_2, \ldots$.
Define $\mathcal{P} = \{p_0, \ldots, p_{n-1}\}$ as the set of all processes.



We now define our notion of valency.
Let $\mathcal{A}$ be a set of executions from some configuration $C$.
Consider any execution $\alpha \in \mathcal{A}$.
A set of processes $\mathcal{P}' \subseteq \mathcal{P}$ is \emph{bivalent in $\alpha$ with respect to $\mathcal{A}$} if, for all $v \in \{0, 1\}$, there exists an execution $\beta$ from $C\alpha$ containing only steps by processes in $\mathcal{P}'$ such that $\alpha\beta \in \mathcal{A}$ and some process has decided $v$ in $C\alpha\beta$.
$\mathcal{P}'$ is \emph{$v$-univalent in $\alpha$ with respect to $\mathcal{A}$} if $\alpha \in \mathcal{A}$ and, for every execution $\beta$ from $C\alpha$ containing only steps by processes in $\mathcal{P}'$ such that $\alpha\beta \in \mathcal{A}$, if some process has decided a value in $C\alpha\beta$, then that value is $v$.
If the set of all processes $\mathcal{P}$ is bivalent (resp. $v$-univalent) in $\alpha$ with respect to $\mathcal{A}$, then we say $\alpha$ is \emph{bivalent (resp. $v$-univalent) with respect to $\mathcal{A}$}.
If $\alpha$ is $v$-univalent with respect to $\mathcal{A}$ for some $v \in \{0, 1\}$, then we say $\alpha$ is \emph{univalent with respect to $\mathcal{A}$}.
If $\mathcal{P'}$ is bivalent (resp. $v$-univalent) in the empty execution from $C$ with respect to $\mathcal{A}$, then we say $\mathcal{P}'$ is \emph{bivalent (resp. $v$-univalent) in $C$ with respect to $\mathcal{A}$}.
If $\mathcal{P}$ is bivalent (resp. $v$-univalent) in $C$ with respect to $\mathcal{A}$, then we say $C$ is \emph{bivalent (resp. $v$-univalent) with respect to $\mathcal{A}$}.



The execution $\alpha$ is \emph{critical with respect to $\mathcal{A}$} if $\alpha$ is bivalent with respect to $\mathcal{A}$, and, for every nonempty execution $\beta$ such that $\alpha\beta \in \mathcal{A}$, $\alpha\beta$ is univalent with respect to $\mathcal{A}$.
If $\alpha$ is critical with respect to $\mathcal{A}$, then we say that process $p_i$ is \emph{on team $v$ in $C\alpha$ with respect to $\mathcal{A}$} if $\alpha p_i$ is $v$-univalent with respect to $\mathcal{A}$.
When $\mathcal{A}$ is clear from context, we will simply say $p_i$ is \emph{on team $v$ in $C\alpha$}.

Notice that we have defined valency and criticality as properties of executions rather than configurations (except that the valency of a configuration is defined with respect to a set of executions from itself).
We do this because not all executions leading to a certain configuration may be in the set of executions considered.
For example, consider a set of executions $\mathcal{A}$ from a configuration $C$.
Suppose that $\alpha \in \mathcal{A}$ is bivalent with respect to $\mathcal{A}$ and let $C' = C\alpha$.
Now suppose there exists some different execution $\beta$ such that $C' = C\beta$.
(Recall that a configuration consists only of the states of all processes and the values of all objects, so multiple different executions can lead from $C$ to $C'$.)
Even though $\alpha$ and $\beta$ both end in the configuration $C'$, it is not necessarily true that $\beta$ is bivalent with respect to $\mathcal{A}$.
In particular, it might true that $\beta \not\in \mathcal{A}$, in which case the valency of $\beta$ with respect to $\mathcal{A}$ is undefined.



\smallskip

Let $C$ be some initial configuration in which some process $p_i$ has input $0$ and another process $p_j$ has input $1$.
Consider an initial configuration $B_0$ in which every process has input $0$ and another initial configuration $B_1$ in which every process has input $1$.
Then $C \widesim{p_i} B_0$, $C \widesim{p_j} B_1$, and all of the objects have the same values in $C$, $B_0$, and $B_1$.
By validity, $p_i$ decides $0$ in its solo-terminating execution from $B_0$ and $p_j$ decides $1$ in its solo-terminating execution from $B_1$.
Therefore, $p_i$ decides $0$ in its solo-terminating execution from $C$ and $p_j$ decides $1$ in its solo-terminating execution from $C$.
Hence, $C$ is bivalent with respect to any set of executions from $C$ that contain the solo-terminating executions of $p_i$ and $p_j$ from $C$.
This gives us the following observation.

\begin{observation}\label{obs:bivinit}
Let $C$ be an initial configuration of an algorithm in which some process $p_i$ has input $0$ and another process $p_j$ has input $1$.
Then $C$ is bivalent with respect to any set of executions from $C$ that contains the solo-terminating executions by $p_i$ and $p_j$ from $C$.
\end{observation}


Consider some set of processes $\mathcal{P}' \subseteq \mathcal{P}$ and an execution $\alpha \in \mathcal{A}$ such that $\mathcal{P}'$ is univalent in $\alpha$ with respect to $\mathcal{A}$.
Let $\beta, \gamma$ be any executions containing only steps by $\mathcal{P}'$ such that $\alpha\beta, \alpha\beta\gamma \in \mathcal{A}$.
If $\mathcal{P}'$ is $v$-univalent in either $\alpha$ or $\alpha\beta$ with respect to $\mathcal{A}$ and some process has decided a value in $C\alpha\beta\gamma$, then that value must be $v$ by definition of $v$-univalent.
Hence, $\mathcal{P}'$ is $v$-univalent is both $\alpha$ and $\alpha\beta$ with respect to $\mathcal{A}$.
We formalize this in the following observation.

\begin{observation}\label{obs:stilluniv}
Let $\mathcal{P}' \subseteq \mathcal{P}$, $\alpha \in \mathcal{A}$, $\alpha\beta \in \mathcal{A}$, and suppose that $\mathcal{P}$ is univalent in $\alpha$ with respect to $\mathcal{A}$ and $\beta$ contains only steps by $\mathcal{P}'$.
Then $\mathcal{P}'$ is $v$-univalent in $\alpha$ with respect to $\mathcal{A}$ if and only if $\mathcal{P}'$ is $v$-univalent in $\alpha\beta$ with respect to $\mathcal{A}$.
\end{observation}

We now define the set of executions we use to perform our valency argument.
As mentioned earlier, in the executions we consider, the maximum number of crashes allowed by any process $p_i$ is proportional to the number of steps by processes $p_0, \ldots, p_{i-1}$.

More formally, for any configuration $C$ and any nonnegative integer $z$, define $\mathcal{E}_z(C)$ as the set of all executions $\alpha$ from $C$ that contain no crashes by $p_0$ and in which, for every process $p_i \in \{p_1, \ldots, p_{n-1}\}$, the number of crashes by $p_i$ is no greater than $zn$ times the number of steps collectively taken by $p_0, \ldots, p_{i-1}$ in $\alpha$.
Define $\mathcal{E}^\star_z(C) \subset \mathcal{E}_z(C)$ as the set of all executions $\alpha$ from $C$ that contain no crashes by $p_0$ and in which, for every process $p_i \in \{p_1, \ldots, p_{n-1}\}$ and every prefix $\alpha'$ of $\alpha$, the number of crashes by $p_i$ is no greater than $zn$ times the number of steps collectively taken by $p_0, \ldots, p_{i-1}$ in $\alpha'$.

Notice that $\mathcal{E}^\star_z(C)$ is prefix-closed but $\mathcal{E}_z(C)$ is not.
For example, if $n = 2$, then $\textit{exec}(C, p_1c_1p_0) \in \mathcal{E}_1(C)$ but $\textit{exec}(p_1c_1p_0) \not\in \mathcal{E}^\star_1(C)$ because the number of crashes by $p_1$ in $\textit{exec}(p_1c_1)$ is greater than $2$ times the number of steps taken by $p_0$ in that execution.


The following two observations follow from the the definitions of $\mathcal{E}_z$ and $\mathcal{E}^\star_z$.

\begin{observation}\label{obs:closedconcat}
For any configuration $C$, any integer $z > 0$, and any $\mathcal{A} \in \{\mathcal{E}_z, \mathcal{E}^\star_z\}$, if $\alpha \in \mathcal{A}(C)$ and $\beta \in \mathcal{A}(C\alpha)$, then $\alpha\beta \in \mathcal{A}(C)$.
\end{observation}

\begin{observation}\label{obs:crashfreesched}
For any configuration $C$, any integer $z > 0$, any $\mathcal{A} \in \{\mathcal{E}_z, \mathcal{E}^\star_z\}$, and any execution or schedule $\sigma$ that contains no crashes, if $\alpha \in \mathcal{A}(C)$, then $\alpha\sigma \in \mathcal{A}(C)$.
\end{observation}

%

Consider a configuration $C$, $\mathcal{A} \in \{\mathcal{E}_z, \mathcal{E}^\star_z\}$, and $\alpha \in \mathcal{A}(C)$ such that $\mathcal{P}' \subseteq \mathcal{P}$ is $v$-univalent in $\alpha$ with respect to $\mathcal{A}(C)$.
Also consider any configuration $C'$ such that $C'$ and $C\alpha$ are indistinguishable to $\mathcal{P}' \subseteq \mathcal{P}$ and all of the objects have the same values in $C'$ and $C\alpha$.
Let $\beta' \in \mathcal{A}(C')$ be an execution that contains only steps by $\mathcal{P}'$.
Then there exists an execution $\beta$ from $C\alpha$ containing only steps by $\mathcal{P}'$ such that $\beta$ and $\beta'$ are indistinguishable to $\mathcal{P}'$.
Furthermore, since $\beta' \in \mathcal{A}(C')$, we must also have $\beta \in \mathcal{A}(C\alpha)$ by the definitions of $\mathcal{E}_z$ and $\mathcal{E}^\star_z$.
By Observation~\ref{obs:closedconcat}, we have $\alpha\beta \in \mathcal{A}(C)$.
Since $\alpha$ is $v$-univalent with respect to $\mathcal{A}(C)$, if any process has decided a value in $C\alpha\beta$, then that value is $v$.
Therefore, if any process has decided a value in $C'\beta'$, then that value is $v$.
Thus, $\mathcal{P}'$ is $v$-univalent in $C'$ with respect to $\mathcal{A}(C')$.
We formalize this below.


\begin{observation}\label{obs:indistuniv}
For any $\mathcal{P}' \subseteq \mathcal{P}$, any configuration $C$, any integer $z > 0$, any $\mathcal{A} \in \{\mathcal{E}_z, \mathcal{E}^\star_z\}$, and any $\alpha \in \mathcal{A}(C)$, if $\mathcal{P}'$ is $v$-univalent in $\alpha$ with respect to $\mathcal{A}(C)$, then for any configuration $C'$ such that $C' \widesim{\mathcal{P}'} C\alpha$ and all of the objects have the same values in $C'$ and $C\alpha$, $\mathcal{P}'$ is $v$-univalent in $C'$ with respect to $\mathcal{A}(C')$.
\end{observation}

In the following lemma, we will show that if $C$ is bivalent with respect to $\mathcal{E}^\star_z(C)$ for some integer $z > 0$, then we can construct an execution $\alpha \in \mathcal{E}^\star_z(C)$ that is critical with respect to $\mathcal{E}^\star_z(C)$, similar to known valency arguments for consensus.
We will also prove a slightly different property: for any $z' \geq z$, there is an execution $\beta \in \mathcal{E}_{z'}(C)$ such that $C\beta$ is bivalent with respect to $\mathcal{E}^\star_{z}(C\beta)$ and, for every execution $\gamma$ such that $\beta\gamma \in \mathcal{E}_{z'}(C)$, the configuration $C\beta\gamma$ is univalent with respect to $\mathcal{E}^\star_z(C\beta\gamma)$.
In other words, it is not possible to infinitely extend an execution in $\mathcal{E}_{z'}(C)$ while keeping the resulting configuration $C'$ bivalent with respect to $\mathcal{E}^\star_z(C')$.
While this is different from our definition of a critical execution, the proof that $\beta \in \mathcal{E}_{z'}(C)$ exists is largely the same as the proof that $\alpha \in \mathcal{E}^\star_z(C)$ exists.
We formally prove the existence of these executions in the following lemma.

\begin{lemma}\label{lem:eventualcrit}
Let $C$ be a configuration that is bivalent with respect to $\mathcal{E}^\star_z(C)$, for some integer $z > 0$.
Then

\begin{enumerate}[label=(\alph*)]
	\item there exists a finite execution $\alpha \in \mathcal{E}^\star_z(C)$ such that $\alpha$ is critical with respect to $\mathcal{E}^\star_z(C)$, and\label{lem:eventualcrit:a}
	\item for any $z' \geq z$, there exists a finite execution $\beta \in \mathcal{E}_{z'}(C)$ such that $C\beta$ is bivalent with respect to $\mathcal{E}^\star_z(C\beta)$ and, for every nonempty finite execution $\gamma$ such that $\beta\gamma \in \mathcal{E}_{z'}(C)$, the configuration $C\beta\gamma$ is univalent with respect to $\mathcal{E}^\star_z(C\beta\gamma)$.\label{lem:eventualcrit:b}
\end{enumerate}
\end{lemma}

\begin{proof}
Suppose that either \ref{lem:eventualcrit:a} or \ref{lem:eventualcrit:b} is false.
We proceed by constructing infinitely many executions $\gamma_0, \gamma_1, \gamma_2, \ldots$ starting from $C$ such that $\gamma_i$ is a proper prefix of $\gamma_{i+1}$ for all $i$.
If \ref{lem:eventualcrit:a} is false, then $\gamma_i \in \mathcal{E}^\star_z(C)$ and $\gamma_i$ is bivalent with respect to $\mathcal{E}^\star_z(C)$ for all $i$.
On the other hand, if \ref{lem:eventualcrit:b} is false, then $\gamma_i \in \mathcal{E}_{z'}(C)$ and $C\gamma_i$ is bivalent with respect to $\mathcal{E}^\star_z(C\gamma_i)$ for all $i$.
In either case, define $\gamma_0$ as the empty execution from $C$.




Now let $i > 0$ and suppose that we have constructed the execution $\gamma_{i-1}$.
If \ref{lem:eventualcrit:a} is false, then the execution $\gamma_{i-1}$ is not critical with respect to $\mathcal{E}^\star_z(C)$.
Since $\gamma_{i-1}$ is bivalent with respect to $\mathcal{E}^\star_z(C)$, there exists a nonempty execution $\delta$ from $C\gamma_{i-1}$ such that $\gamma_{i-1}\delta \in \mathcal{E}^\star_z(C)$ and $\gamma_{i-1}\delta$ is bivalent with respect to $\mathcal{E}^\star_z(C)$.
Define $\gamma_i = \gamma_{i-1}\delta$.

Otherwise, \ref{lem:eventualcrit:b} is false.
Then $C\gamma_{i-1}$ is bivalent with respect to $\mathcal{E}^\star_z(C\gamma_{i-1})$ and $\gamma_{i-1} \in \mathcal{E}_{z'}(C)$.
Hence, by the negation of \ref{lem:eventualcrit:b}, there exists a nonempty execution $\delta$ from $C\gamma_{i-1}$ such that $\gamma_{i-1}\delta \in \mathcal{E}_{z'}(C)$ and $C\gamma_{i-1}\delta$ is bivalent with respect to $\mathcal{E}^\star_z(C\gamma_{i-1}\delta)$.
Define $\gamma_i = \gamma_{i-1}\delta$.
This completes the construction.


Since $\gamma_i$ is a proper prefix of $\gamma_{i+1}$, for all $i$, the executions $\gamma_0, \gamma_1, \gamma_2, \ldots$ are strictly increasing in length.
For all $i$, define $\alpha_{i}$ as the nonempty execution from $C\gamma_i$ to $C\gamma_{i+1}$, that is, $\gamma_{i+1} = \gamma_i\alpha_{i}$.
More generally, for any $j > i$, we have $\gamma_j = \gamma_i\alpha_i\alpha_{i+1}\ldots\alpha_{j-1}$.
Since each $\alpha_i$ is nonempty, there must exist a process $p_k$ such that, for every $i \geq 0$, there exists a $j \geq i$ such that $p_k$ either takes a step or crashes in $\alpha_{j}$.
Let $p_k$ be the process with the minimum identifier for which this property holds.
Then there is an $m \geq 0$ such that $p_0, \ldots, p_{k-1}$ neither take a step nor crash in $\alpha_{m'}$, for all $m' \geq m$.
Let $t$ be the number of steps taken collectively by $p_0, \ldots, p_{k-1}$ in $\gamma_{m}$, and let $f$ be the number of times that $p_k$ crashes in $\gamma_m$.
Since $z' \geq z$, we know that $\mathcal{E}^\star_z(C) \subset \mathcal{E}_{z'}(C)$ for all $C$, so $\gamma_m, \gamma_{m+1}, \ldots \in \mathcal{E}_{z'}(C)$. 
Hence, process $p_k$ crashes no more than $z'nt - f$ times in $\alpha_{m}\alpha_{m+1}\ldots$.
Therefore, there exists an $\ell \geq m$ such that $p_k$ does not crash in $\alpha_\ell\alpha_{\ell+1}\ldots$.
Hence, $p_k$ takes infinitely many steps without crashing in $\alpha_\ell\alpha_{\ell+1}\ldots$.
By recoverable wait-freedom, there exists an $r \geq \ell$ such that $p_k$ has decided in $C\gamma_r$.
However, in case \ref{lem:eventualcrit:a} is false, $\gamma_r$ is bivalent with respect to $\mathcal{E}^\star_z(C)$, and in case \ref{lem:eventualcrit:b} is false, $\gamma_r$ is bivalent with respect to $\mathcal{E}^\star_z(C\gamma_r)$.
In either case, there exists an execution $\beta_{\bar{v}}$ from $C\gamma_r$ such that some process has decided $\bar{v}$ in $C\gamma_r\beta_{\bar{v}}$.
This contradicts agreement.

\end{proof}

%


We will now prove some useful properties about critical executions.
For the following five lemmas, let $C$ be a configuration, let $z \geq 1$ be an integer, and let $\alpha \in \mathcal{E}_z^\star(C)$ such that $\alpha$ is critical with respect to $\mathcal{E}^\star_z(C)$.
We proceed by proving that there is at least one process on each team in $C\alpha$.

\begin{lemma}\label{lem:diffteams}
There is at least one process on team $0$ in $C\alpha$ and at least one process on team $1$ in $C\alpha$.
\end{lemma}

\begin{proof}
Observation~\ref{obs:crashfreesched} implies that $\alpha p_0 \in \mathcal{E}^\star_z(C)$.
Hence, $p_0$ is on some team $v \in \{0, 1\}$ in $C\alpha$.
Then $\alpha p_0$ is $v$-univalent with respect to $\mathcal{E}^\star_z(C)$.
Let $\delta$ be $p_0$'s solo-terminating execution from $C\alpha p_0$.
Observation~\ref{obs:crashfreesched} implies that $\alpha p_0\delta \in \mathcal{E}^\star_z(C)$, so $p_0$ has decided $v$ in $C\alpha p_0\delta$.

Consider any $i \in \{1, \ldots, n-1\}$ such that $\alpha c_i \in \mathcal{E}^\star_z(C)$.
Since $C\alpha$ is critical with respect to $\mathcal{E}^\star_z(C)$, the configuration $C\alpha c_i$ is univalent with respect to $\mathcal{E}^\star_z(C)$.
Since $\alpha c_i \in \mathcal{E}^\star_z(C)$, we also have $\alpha c_ip_0 \in \mathcal{E}^\star_z(C)$ by Observation~\ref{obs:crashfreesched}.
Hence, $C\alpha c_ip_0$ is also univalent with respect to $\mathcal{E}^\star_z(C)$ by Observation~\ref{obs:stilluniv}.

Notice that $C\alpha p_0 \widesim{p_0} C\alpha c_ip_0$ and all of the objects have the same values in these two configurations.
Hence, $p_0$ decides $v$ in its solo-terminating execution from $C\alpha c_ip_0$.
Equivalently, $p_0$ decides $v$ in its solo-terminating execution from $C\alpha c_i$.
Let $\gamma$ be $p_0$'s solo-terminating execution from $C\alpha c_i$.
By Observation~\ref{obs:crashfreesched}, $\alpha c_i\gamma \in \mathcal{E}_z^\star(C)$.
Since $p_0$ decides $v$ in $\gamma$ and $\alpha c_i$ is univalent with respect to $\mathcal{E}_z^\star(C)$, Observation~\ref{obs:stilluniv} implies that $\alpha c_i$ is $v$-univalent with respect to $\mathcal{E}_z^\star(C)$.

Since $\alpha$ is critical with respect to $\mathcal{E}^\star_z(C)$, it is bivalent with respect to $\mathcal{E}^\star_z(C)$.
Hence, there exists an execution $\beta_{\bar{v}}$ such that $\alpha\beta_{\bar{v}} \in \mathcal{E}^\star_z(C)$ and some process has decided $\bar{v}$ in $C\alpha\beta_{\bar{v}}$.
Since $\alpha p_0$ and $\alpha c_i$ are $v$-univalent with respect to $\mathcal{E}^\star_z(C)$ for all $i \in \{1, \ldots, n-1\}$ such that $\alpha c_i \in \mathcal{E}^\star_z(C)$, the execution $\beta_{\bar{v}}$ begins with a step by some process $p_j \neq p_0$.
Hence, $p_j$ is on team $\bar{v}$ in $C\alpha$.
\end{proof}

Notice that, since $\alpha$ is critical with respect to $\mathcal{E}^\star_z(C)$, we know that $\alpha$ is bivalent with respect to $\mathcal{E}^\star_z(C)$ by definition of critical.
It takes slightly more effort to show that $C\alpha$ is bivalent with respect to $\mathcal{E}^\star_z(C\alpha)$, which we do in the following lemma.


\begin{lemma}\label{lem:critisbiv}
$C\alpha$ is bivalent with respect to $\mathcal{E}^\star_z(C\alpha)$.
\end{lemma}

\begin{proof}
By Lemma~\ref{lem:diffteams}, there is a process $p_i$ on team $0$ in $C\alpha$ and a process $p_j$ on team $1$ in $C\alpha$.
Let $\delta_i$ be $p_i$'s solo-terminating execution from $C\alpha p_i$, and let $\delta_j$ be $p_j$'s solo-terminating execution from $C\alpha p_j$.
Since $p_i$ is on team $0$ in $C\alpha$ and $p_j$ is on team $1$ in $C\alpha$, the execution $\alpha p_i$ is $0$-univalent with respect to $\mathcal{E}^\star_z(C)$ and $\alpha p_j$ is $1$-univalent with respect to $\mathcal{E}^\star_z(C)$.
Observation~\ref{obs:crashfreesched} implies that $\alpha p_i\delta_i \in \mathcal{E}^\star_z(C)$ and $\alpha p_j\delta_j \in \mathcal{E}^\star_z(C)$.
Since $\alpha p_i$ is $0$-univalent with respect to $\mathcal{E}^\star_z(C)$ and $\alpha p_j$ is $1$-univalent with respect to $\mathcal{E}^\star_z(C)$, $p_i$ has decided $0$ in $C\alpha p_i\delta_i$ and $p_j$ has decided $1$ in $C\alpha p_j\delta_j$.
Observation~\ref{obs:crashfreesched} implies that $\delta_i, \delta_j \in \mathcal{E}^\star_z(C\alpha)$.
Hence, $C\alpha$ is bivalent with respect to $\mathcal{E}^\star_z(C\alpha)$.
\end{proof}

We now prove that every process is poised to access the same object in $C\alpha$.
The proof of this is essentially identical to typical proofs of the same property for wait-free consensus algorithms.

\begin{lemma}\label{lem:sameobject}
Every process is poised to apply an operation to the same object in $C\alpha$.
\end{lemma}

\begin{proof}
By Lemma~\ref{lem:diffteams}, there is a process that is on team $0$ in $C\alpha$ and there is a process that is on team $1$ in $C\alpha$.
Let $p_i$ be any process that is on team $0$ in $C\alpha$ and let $p_j$ be any process that is on team $1$ in $C\alpha$.
If $p_i$ and $p_j$ are poised to access different objects in $C\alpha$, then $C\alpha p_ip_j = C\alpha p_jp_i$.
Hence, $p_i$ decides the same value in its solo-terminating executions from $C\alpha p_ip_j$ and $C\alpha p_jp_i$.
Since $\alpha p_i$ is $0$-univalent with respect to $\mathcal{E}^\star_z(C)$ and $\alpha p_i p_j \in \mathcal{E}^\star_z(C)$ by Observation~\ref{obs:crashfreesched}, $\alpha p_ip_j$ is also $0$-univalent with respect to $\mathcal{E}^\star_z(C)$ by Observation~\ref{obs:stilluniv}.
Hence, $p_i$ decides $0$ in its solo-terminating execution from $C\alpha p_i p_j$.
Therefore, $p_i$ also decides $0$ in its solo-terminating execution from $C\alpha p_j p_i$.
However, since $\alpha p_j$ is $1$-univalent with respect to $\mathcal{E}^\star_z(C)$ and $\alpha p_jp_i \in \mathcal{E}^\star_z(C)$ by Observation~\ref{obs:crashfreesched}, the execution $\alpha p_jp_i$ is also $1$-univalent with respect to $\mathcal{E}^\star_z(C)$ by Observation~\ref{obs:stilluniv}.
This is a contradiction.
\end{proof}

For the following two lemmas, let $O$ be the object that every process is poised to access in $C\alpha$.
Such an object exists by Lemma~\ref{lem:sameobject}.

In the next lemma we will show that, if there are two nonempty schedules $\sigma, \tau \in \mathcal{S}(\mathcal{P})$ beginning with processes on different teams such that $\sigma$ and $\tau$ change $O$ to the same value when applied in $C\alpha$, then either $\sigma = p_{n-1}$ or $\tau = p_{n-1}$.
To see why this is true, notice that $p_{n-1}$ cannot distinguish between $C\alpha\sigma c_{n-1}$ and $C\alpha\tau c_{n-1}$ and all of the objects have the same values in these two configurations, so $p_{n-1}$ decides the same value in its solo-terminating executions from these two configurations.
If neither $\sigma = p_{n-1}$ nor $\tau = p_{n-1}$, then both $\sigma$ and $\tau$ contain a process with a lower identifier than $p_{n-1}$.
Since $z \geq 1$, this implies that $\alpha\sigma c_{n-1}, \alpha\tau c_{n-1} \in \mathcal{E}^\star_z(C)$.
Hence, $\alpha\sigma c_{n-1}$ and $\alpha\tau c_{n-1}$ have different valencies with respect to $\mathcal{E}^\star_z(C)$.
This leads to a contradiction.
We formalize this in the following lemma.

\begin{lemma}\label{lem:samevalrestrict}
Suppose that $p_{n-1}$ is on team $\bar{v}$ in $C\alpha$.
If there is a process $p_i$ that is on team $v$ in $C\alpha$, a process $p_j$ that is on team $\bar{v}$ in $C\alpha$, an $R_i \in \mathcal{S}\bigl(\mathcal{P} - \{p_i\}\bigr)$, and an $R_j \in \mathcal{S}\bigl(\mathcal{P} - \{p_j\}\bigr)$ such that $\textit{value}(O, C\alpha p_iR_i) = \textit{value}(O, C\alpha p_jR_j)$, then $p_j = p_{n-1}$ and $R_j = \langle\rangle$.
\end{lemma}

\begin{proof}
To obtain a contradiction, suppose that $p_j \neq p_{n-1}$ or $R_j \neq \langle\rangle$.
Notice that $C\alpha p_iR_i c_{n-1} \widesim{p_{n-1}} C\alpha p_jR_j c_{n-1}$ and all of the objects have the same values in $C\alpha p_iR_ic_{n-1}$ and $C\alpha p_jR_jc_{n-1}$.
Hence, $p_{n-1}$ decides the same value in its solo-terminating executions from $C\alpha p_iR_i c_{n-1}$ and $C\alpha p_jR_j c_{n-1}$.
Since $z \geq 1$ and the schedules $p_iR_i$ and $p_jR_j$ both contain a process with a smaller identifier than $p_{n-1}$, we have $\textit{exec}(C\alpha, p_iR_i c_{n-1}), \textit{exec}(C\alpha, p_jR_j c_{n-1}) \in \mathcal{E}^\star_z(C\alpha)$.
Since $\alpha \in \mathcal{E}^\star_z(C)$, Observation~\ref{obs:closedconcat} implies that $\alpha p_iR_ic_{n-1}, \alpha p_jR_jc_{n-1} \in \mathcal{E}^\star_z(C)$.
Since $\alpha p_i$ is $v$-univalent with respect to $\mathcal{E}^\star_z(C)$ and $\alpha p_j$ is $\bar{v}$-univalent with respect to $\mathcal{E}^\star_z(C)$, it must be true that $\alpha p_iR_ic_{n-1}$ is $v$-univalent with respect to $\mathcal{E}^\star_z(C)$ and $\alpha p_jR_jc_{n-1}$ is $\bar{v}$-univalent with respect to $\mathcal{E}^\star_z(C)$.
Hence, $p_{n-1}$ decides different values in its solo-terminating executions from $C\alpha p_iR_ic_{n-1}$ and $C\alpha p_jR_jc_{n-1}$.
This is a contradiction.
\end{proof}

For all $v \in \{0, 1\}$, define $T_v$ as the set of all processes on team $v$ in $C\alpha$.
Furthermore, define $U_v = \bigl\{\textit{value}(O, C\alpha\sigma) : \sigma = p_{i_v}\sigma',\: p_{i_v} \in T_v,\: \sigma' \in \mathcal{S}\bigl(\mathcal{P} - \{p_{i_v}\}\bigr)\bigr\}$.
The configuration $C\alpha$ is \emph{$v$-hiding} if $U_0 \cap U_1 = \emptyset$ and \emph{value}$(O, C\alpha) \in U_v$.
Intuitively, $C\alpha$ is $v$-hiding if it is not possible for the two teams to change $O$ to the same value, but team $v$ can be hidden in the sense that a schedule beginning with a step by a process on team $v$ can result in $O$ having the same value that it did in $C\alpha$.

The configuration $C\alpha$ is \emph{$n$-recording} if $U_0 \cap U_1 = \emptyset$ and, if \emph{value}$(O, C\alpha) \in U_v$ for some $v \in \{0, 1\}$, then $|T_{\bar{v}}| = 1$.
Notice that, if there exists an $n$-recording configuration in which all of the processes are poised to access an object $O$ of type $\mathcal{T}$, then $\mathcal{T}$ is $n$-recording.
For any $k \in \{1, \ldots, n-1\}$, let $\lambda_k$ denote the schedule $c_kc_{k+1}\ldots c_{n-1}$.
The next observation follows from the definitions of $v$-hiding and $n$-recording.

\begin{observation}\label{obs:cases}
$C\alpha$ is either $n$-recording, $v$-hiding for some $v \in \{0, 1\}$, or there exists a process $p_i$ on team $0$ in $C\alpha$, a process $p_j$ on team $1$ in $C\alpha$, an $R_i \in \mathcal{S}\bigl(\mathcal{P} - \{p_i\}\bigr)$, and an $R_j \in \mathcal{S}\big(\mathcal{P} - \{p_j\}\bigr)$ such that $\textit{value}(O, C\alpha p_iR_i) = \textit{value}(O, C\alpha p_jR_j)$.
\end{observation}

For the following lemma, suppose that there exists a process $p_i$ on team $v$ in $C\alpha$ and an $R_i \in \mathcal{S}\bigl(\mathcal{P} - \{p_i\}\bigr)$ such that $\textit{value}(O, C\alpha p_iR_i) = \textit{value}(O, C\alpha X)$, where either $X = p_{n-1}$ and $p_{n-1}$ is on team $\bar{v}$ in $C\alpha$ or $X = \langle\rangle$.
Notice that if $X = \langle\rangle$, then $C\alpha$ is $v$-hiding.
Let $p_{\ell}$ be the process with the smallest identifier among $p_i$ and the processes in $R_i$.
In the following lemma, we show that if $\ell < n-1$ and the processes $p_{\ell+1}, \ldots, p_{n-1}$ crash in $C\alpha X$, and then these same processes take steps to produce an execution that is critical with respect to $\mathcal{E}^\star_z(C\alpha X\lambda_{\ell+1})$, then the resulting configuration is either $n$-recording or $v$-hiding.

To provide some intuition for this, notice that $\alpha p_iR_i\lambda_{\ell+1} \in \mathcal{E}^\star_z(C)$ because $p_\ell \in p_iR_i$ has a smaller identifier than any of the processes that crashed in $\lambda_{\ell+1}$.
Therefore, $\alpha p_iR_i\lambda_{\ell+1}$ is $v$-univalent with respect to $\mathcal{E}^\star_z(C)$.
For this example, suppose that $X = \langle\rangle$ and, to obtain a contradiction, suppose that there exists an execution $\alpha'$ containing only events by $p_{\ell+1}, \ldots, p_{n-1}$ such that $\alpha'$ is critical with respect to $\mathcal{E}^\star_z(C\alpha \lambda_{\ell+1})$ and $C\alpha \lambda_{\ell+1}\alpha'$ is $\bar{v}$-hiding.
Since $C\alpha p_iR_i\lambda_{\ell+1}$ and $C\alpha \lambda_{\ell+1}$ are indistinguishable to $p_{\ell+1}, \ldots, p_{n-1}$ and all of the objects have the same values in these two configurations, $\{p_{\ell+1}, \ldots, p_{n-1}\}$ is $v$-univalent in $C\alpha \lambda_{\ell+1}$ with respect to $\mathcal{E}^\star_z(C\alpha \lambda_{\ell+1})$ by Observation~\ref{obs:indistuniv}.
We can use this, combined with the fact that $\alpha'$ contains only events by $p_{\ell+1}, \ldots, p_{n-1}$, to show that $p_{\ell+1}, \ldots, p_{n-1}$ are on team $v$ in $C\alpha \lambda_{\ell+1}\alpha'$ with respect to $\mathcal{E}_z^\star(C\alpha\lambda_{\ell+1})$.
Define $C' = C\alpha \lambda_{\ell+1}\alpha'$.
Since $C'$ is $\bar{v}$-hiding, there exists a process $p_j$ on team $\bar{v}$ in $C'$ and an $R_j \in \mathcal{S}\bigl(\mathcal{P} - \{p_j\}\bigr)$ such that $\textit{value}(O, C'p_jR_j) = \textit{value}(O, C')$.
Since $p_{n-1}$ is on team $v$ in $C'$, we know that $p_j \neq p_{n-1}$.
Therefore, we can show that $\alpha'p_jR_jc_{n-1} \in \mathcal{E}^\star_z(C\alpha\lambda_{\ell+1})$.
Hence, $p_{n-1}$ decides $\bar{v}$ in its solo-terminating execution from $C' p_jR_j c_{n-1}$.
Since $p_{n-1}$ cannot distinguish between $C' p_jR_jc_{n-1}$ and $C' c_{n-1}$ and all of the objects have the same values in these configurations, $p_{n-1}$ decides $\bar{v}$ in its solo-terminating execution from $C'c_{n-1}$ as well.

Since the configurations $C\alpha\lambda_{\ell+1}$ and $C\alpha p_iR_i\lambda_{\ell+1}$ are indistinguishable to $p_{\ell+1}, \ldots, p_{n-1}$, all of the objects have the same values in these two configurations, and $\alpha'$ only contains events by $p_{\ell+1}, \ldots, p_{n-1}$, there exists an execution $\gamma$ from $C\alpha p_iR_i\lambda_{\ell+1}$ that is indistinguishable from $\alpha'$ to $p_{\ell+1}, \ldots, p_{n-1}$.
Hence, $p_{n-1}$ decides $\bar{v}$ in its solo-terminating execution from $C\alpha p_iR_i\lambda_{\ell+1}\gamma c_{n-1}$.
However, since $p_\ell$ takes a step in $p_iR_i$, $\ell < n-1$, $z \geq 1$, and $n \geq 2$, we can use the definition of $\mathcal{E}_z^\star(C)$ to show that $\alpha p_iR_i\lambda_{\ell+1}\gamma c_{n-1} \in \mathcal{E}^\star_z(C)$.
Since $\alpha p_iR_i\lambda_{\ell+1}$ is $v$-univalent with respect to $\mathcal{E}^\star_z(C)$, process $p_{n-1}$ must decide $v$ in its solo-terminating execution from $C\alpha p_iR_i\lambda_{\ell+1}\gamma c_{n-1}$.
This is a contradiction.
\ifarxiv
We formalize this argument and consider the other possible cases in the proof of the following lemma.
\fi


\begin{lemma}\label{lem:maintech}
For any $k \in \{\ell+1, \ldots, n-1\}$ and any execution $\alpha' \in \mathcal{E}^\star_z(C\alpha X\lambda_k)$ that contains no events by $p_0, \ldots, p_{k-1}$, if $\alpha'$ is critical with respect to $\mathcal{E}^\star_z(C'X\lambda_k)$, then $C\alpha X\lambda_k\alpha'$ is either $n$-recording or $v$-hiding.
\end{lemma}
\ifarxiv
\begin{proof}
Notice that the lemma vacuously holds if $\ell = n-1$.
To obtain a contradiction, suppose that there exists some $\ell < n-1$ for which the lemma does not hold.
Then there exists $k \in \{\ell+1, \ldots, n-1\}$ and an execution $\alpha' \in \mathcal{E}^\star_z(C\alpha X\lambda_k)$ that contains no events by $p_0, \ldots, p_{k-1}$ such that $\alpha'$ is critical with respect to $\mathcal{E}^\star_z(C\alpha X\lambda_k)$ and $C\alpha X\lambda_k\alpha'$ is neither $n$-recording nor $v$-hiding.

Since $p_i$ is on team $v$ in $C\alpha$, we know that $\alpha p_iR_i$ is $v$-univalent with respect to $\mathcal{E}^\star_z(C)$.
Furthermore, since $p_\ell$ takes a step in $p_iR_i$, $z \geq 1$, $n \geq 2$, and $\ell < k$, we know that $\textit{exec}(C\alpha, p_iR_i\lambda_k) \in \mathcal{E}^\star_z(C\alpha)$ by definition of $\mathcal{E}_z^\star(C\alpha)$.
Since $\alpha \in \mathcal{E}^\star_z(C)$, Observation~\ref{obs:closedconcat} implies that $\alpha p_iR_i\lambda_k \in \mathcal{E}^\star_z(C)$.
Hence, $\alpha p_iR_i\lambda_k$ is $v$-univalent with respect to $\mathcal{E}^\star_z(C)$ by Observation~\ref{obs:stilluniv}.
In particular, $\{p_k, \ldots, p_{n-1}\}$ is $v$-univalent in $\alpha p_iR_i\lambda_k$ with respect to $\mathcal{E}^\star_z(C)$.
Notice that $C\alpha p_iR_i\lambda_k$ and $C\alpha X\lambda_k$ are indistinguishable to $p_k, \ldots, p_{n-1}$ and all of the objects have the same values in these two configurations.
By Observation~\ref{obs:indistuniv}, $\{p_k, \ldots, p_{n-1}\}$ is $v$-univalent in $C\alpha X\lambda_k$ with respect to $\mathcal{E}^\star_z(C\alpha X\lambda_k)$.
Since $\alpha' \in \mathcal{E}^\star_z(C\alpha X\lambda_k)$ and $\alpha'$ only contains steps by $p_k, \ldots, p_{n-1}$, Observation~\ref{obs:stilluniv} implies that $\{p_k, \ldots, p_{n-1}\}$ is $v$-univalent in $\alpha'$ with respect to $\mathcal{E}^\star_z(C\alpha X\lambda_k)$.
In particular, $p_{n-1}$ is on team $v$ in $C\alpha X\lambda_k\alpha'$.\todo{might need some more detail}

Since $C\alpha X\lambda_k$ and $C\alpha p_iR_i\lambda_k$ are indistinguishable to $p_k, \ldots, p_{n-1}$, all of the objects have the same values in these two configurations, and $\alpha'$ contains only events by $p_k, \ldots, p_{n-1}$, there exists an execution $\gamma$ from $C\alpha p_iR_i\lambda_k$ containing only events by $p_k, \ldots, p_{n-1}$ that is indistinguishable from $\alpha'$ to $p_k, \ldots, p_{n-1}$.
Furthermore, since the schedules of $\alpha'$ and $\gamma$ consisting of the events by $p_{k}, \ldots, p_{n-1}$ are the same,  $\gamma \in \mathcal{E}^\star_z(C\alpha p_iR_i\lambda_k)$.

Since $\alpha p_iR_i\lambda_k \in \mathcal{E}^\star_z(C)$ and $\gamma \in \mathcal{E}^\star_z(C\alpha p_iR_i\lambda_k)$, Observation~\ref{obs:closedconcat} implies that $\alpha p_iR_i\lambda_k\gamma \in \mathcal{E}^\star_z(C)$.
By the definition of $\mathcal{E}^\star_z(C)$, for every process $p_r$ and every prefix $\gamma'$ of $\alpha p_iR_i\lambda_k\gamma$, the number of crashes by $p_r$ is no greater than $zn$ times the number of steps collectively taken by $p_0, \ldots, p_{r-1}$ in $\gamma'$.
We now show that $\alpha p_iR_i\lambda_k \gamma c_{n-1} \in \mathcal{E}^\star_z(C)$.
To do so, it suffices to prove that the number of crashes by $p_{n-1}$ in $\alpha p_iR_i\lambda_k\gamma c_{n-1}$ is no greater than $zn$ times the number of steps collectively taken by $p_0, \ldots, p_{n-2}$ in $\alpha p_iR_i\lambda_k\gamma c_{n-1}$.
To see why this is true, first notice that $\alpha \in \mathcal{E}^\star_z(C)$ and $\gamma \in \mathcal{E}^\star_z(C\alpha p_iR_i\lambda_k)$.
Hence, if $s$ is the number of steps taken by $p_0, \ldots, p_{n-2}$ in $\alpha$ and $\gamma$, then the number of crashes by $p_{n-1}$ in $\alpha$ and $\gamma$ is at most $zns$.
Furthermore, since $p_\ell \in p_iR_i$, the processes $p_0, \ldots, p_{n-2}$ take at least $s+1$ steps and $p_{n-1}$ crashes at most $zns + 2$ times in $\alpha p_iR_i\lambda_k\gamma c_{n-1}$.
Since $z \geq 1$ and $n \geq 2$, we have $zns + 2 \leq zns + zn = zn(s+1)$.
Therefore, $\alpha p_iR_i\lambda_k\gamma c_{n-1} \in \mathcal{E}^\star_z(C)$.

Since $\alpha p_iR_i$ is $v$-univalent with respect to $\mathcal{E}^\star_z(C)$, $\alpha p_iR_i\lambda_k\gamma c_{n-1}$ is also $v$-univalent with respect to $\mathcal{E}^\star_z(C)$ by Observation~\ref{obs:stilluniv}.
Hence, $p_{n-1}$ decides $v$ in its solo-terminating execution from $C\alpha p_iR_i\lambda_k\gamma c_{n-1}$.

\medskip

Define $C' = C\alpha X\lambda_k\alpha'$.
Since $C'$ is neither $n$-recording nor $v$-hiding by assumption, Observation~\ref{obs:cases} implies that either $C'$ is $\bar{v}$-hiding or there exists a process $p_j$ on team $v$ in $C'$, a process $p_m$ on team $\bar{v}$ in $C'$, an $R_j \in \mathcal{S}\bigl(\mathcal{P} - \{p_j\}\bigr)$, and an $R_m \in \mathcal{S}\bigl(\mathcal{P} - \{p_m\}\bigr)$ such that $\textit{value}(O, C'p_jR_j) = \textit{value}(O, C'p_mR_m)$.
We proceed by considering these two cases separately.

In the first case, $C'$ is $\bar{v}$-hiding.
Then there exists a process $p_m$ on team $\bar{v}$ in $C'$  and a schedule $R_m \in \mathcal{S}\bigl(\mathcal{P} - \{p_m\}\bigr)$ such that \emph{value}$(O, C' p_mR_m) =$ \emph{value}$(O, C')$.
Since $p_{n-1}$ is on team $v$ in $C'$, we know that $p_m \neq p_{n-1}$, that is, $m < n-1$.
Hence, $p_mR_mc_{n-1} \in \mathcal{E}^\star_z(C')$.
Since $\alpha' \in \mathcal{E}^\star_z(C\alpha X\lambda_k)$, Observation~\ref{obs:closedconcat} implies that $\alpha' p_mR_mc_{n-1} \in \mathcal{E}^\star_z(C\alpha X\lambda_k)$.
Since $p_m$ is on team $\bar{v}$ in $C'$, $\alpha' p_mR_mc_{n-1}$ is $\bar{v}$-univalent with respect to $\mathcal{E}^\star_z(C\alpha X\lambda_k)$.
Then $p_{n-1}$ decides $\bar{v}$ in its solo-terminating execution from $C'p_mR_mc_{n-1}$.
Since all of the objects have the same values in $C'c_{n-1}$ and $C'p_mR_mc_{n-1}$ and these two configurations are indistinguishable to $p_{n-1}$, we know that $p_{n-1}$ decides $\bar{v}$ in its solo-terminating execution from $C'c_{n-1}$ as well.
Furthermore, $C'c_{n-1}$ and $C\alpha p_iR_i\lambda_k\gamma c_{n-1}$ are indistinguishable to $p_{n-1}$ and all of the objects have the same values in these configurations.
Therefore, $p_{n-1}$ decides $\bar{v}$ in its solo-terminating execution from $C\alpha p_iR_i\lambda_k\gamma c_{n-1}$ as well, which is a contradiction.

\medskip

Now suppose that there exists a process $p_j$ on team $v$ in $C'$, a process $p_m$ on team $\bar{v}$ in $C'$, an $R_j \in \mathcal{S}\bigl(\mathcal{P} - \{p_j\}\bigr)$, and an $R_m \in \mathcal{S}\bigl(\mathcal{P} - \{p_m\}\bigr)$ such that $\textit{value}(O, C' p_jR_j) = \textit{value}(O, C' p_mR_m)$.
Recall that $p_{n-1}$ is on team $v$ in $C'$.
By Lemma~\ref{lem:samevalrestrict}, $p_j = p_{n-1}$ and $R_j = \langle\rangle$.
Hence, $\textit{value}(O, C' p_{n-1}) = \textit{value}(O, C' p_mR_m)$.
Thus, all of the objects have the same values in $C' p_{n-1}c_{n-1}$ and $C' p_mR_mc_{n-1}$ and these two configurations are indistinguishable to $p_{n-1}$.
Since $p_m \neq p_{n-1}$, we have $m < n-1$.
Hence, $\textit{exec}(C', p_mR_mc_{n-1}) \in \mathcal{E}^\star_z(C')$.
Since $\alpha' \in \mathcal{E}^\star_z(C\alpha X\lambda_k)$, Observation~\ref{obs:closedconcat} implies that $\alpha'p_mR_mc_{n-1} \in \mathcal{E}^\star_z(C\alpha X\lambda_k)$.
By Observation~\ref{obs:stilluniv}, $\alpha'p_mR_mc_{n-1}$ is $\bar{v}$-univalent with respect to $\mathcal{E}^\star_z(C\alpha X\lambda_k)$.
Hence, $p_{n-1}$ decides $\bar{v}$ in its solo-terminating executions from $C'X\lambda_k\alpha' p_mR_mc_{n-1}$ and $C\alpha X\lambda_k\alpha' p_{n-1}c_{n-1}$.
This implies that $p_{n-1}$ decides $\bar{v}$ in its solo-terminating execution from $C\alpha p_iR_i\lambda_k\gamma p_{n-1}c_{n-1}$ as well.

Since $\alpha p_iR_i\lambda_k \gamma c_{n-1} \in \mathcal{E}^\star_z(C)$, we know that $\alpha p_iR_i\lambda_k \gamma p_{n-1} c_{n-1} \in \mathcal{E}^\star_z(C)$ as well.
Since $\alpha p_iR_i$ is $v$-univalent with respect to $\mathcal{E}^\star_z(C)$, the execution $\alpha p_iR_i\lambda_k\gamma p_{n-1}c_{n-1}$ is also $v$-univalent with respect to $\mathcal{E}^\star_z(C)$ by Observation~\ref{obs:stilluniv}.
Therefore, $p_{n-1}$ decides $v$ in its solo-terminating execution from $C\alpha p_iR_i\lambda_k\gamma p_{n-1}c_{n-1}$.
This is a contradiction.
Hence, $C' = C\alpha X\lambda_k\alpha'$ is $v$-hiding.
\end{proof}
\fi

We will now prove our main theorem.
To do so, we will consider an arbitrary recoverable wait-free consensus algorithm using objects with types $\mathcal{T}_0, \mathcal{T}_1, \mathcal{T}_2, \ldots$, and we will apply Lemma~\ref{lem:maintech} to inductively construct a chain of $v$-hiding configurations $D_1', \ldots, D_{\ell-1}'$ followed by an $n$-recording configuration $D_\ell'$.
Every configuration $D_{i}'$ is reachable from some configuration $D_i$ via an execution $\alpha_i$ that is critical with respect to $\mathcal{E}^\star_1(D_{i})$.
The main idea of the argument is to force longer and longer suffixes of $p_0, \ldots, p_{n-1}$ to be on team $v$ in the $v$-hiding configurations, until eventually we either reach an $n$-recording configuration in which some process $p_i$ with $i > 0$ is on team $\bar{v}$, or we reach a configuration that is $v$-hiding and in which only $p_0$ is on team $\bar{v}$.
Such a configuration is also $n$-recording.

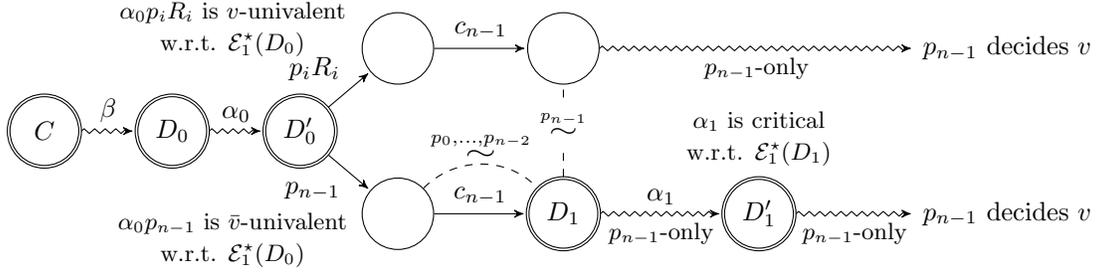
\begin{figure*}[t]
\centering
\begin{tikzpicture}[shorten >=1pt,node distance=2.9cm,on grid,auto] 
	\tikzset{every state/.style={minimum size=0.95cm}}

	\node[state,accepting]		(C)		{$C$};
	\node[state,accepting]		(D0) [right=1.7cm of C]	{$D_0$};
	\node[state,accepting]		(D0') [right=1.7cm of D0] 	{$D_0'$};

	\node[state]		(pn-1) [below right=1.1cm and 1.3cm of D0']		{};
	\node[state,accepting]		(D1)	[right=2.2cm of pn-1] {$D_1$};
	\node[state,accepting]		(D1') [right=2.6cm of D1] {$D_1'$};
	\node		(D1'pn-1)		[right=3.3cm of D1'] {$p_{n-1}$ decides $v$};
	
	\node[state]		(piRi) [above right=1.1cm and 1.3cm of D0']		{};
	\node[state]		(piRicn-1) [right=2.2cm of piRi] {};
	\node		(piRipn-1)		[above=2.2cm of D1'pn-1] {$p_{n-1}$ decides $v$};

	\draw [-,dashed] (pn-1) to [bend left=45] (D1);
   	\node[above left=0.9cm and 1.1cm of D1] {$\widesim{p_0, \ldots, p_{n-2}}$};
   	
   	\draw [-,dashed] (D1) to (piRicn-1);
   	\node[fill=white,above=1.2cm of D1,align=center] {\small $\widesim{p_{n-1}}$};
   	
   	\node[above left=0.25cm and 2.2cm of piRi,align=center] {\small $\alpha_0p_iR_i$ is $v$-univalent \\ \small w.r.t. $\mathcal{E}^\star_1(D_0)$};
   	
   	\node[below left=0.35cm and 2.2cm of pn-1,align=center] {\small $\alpha_0p_{n-1}$ is $\bar{v}$-univalent \\ \small w.r.t. $\mathcal{E}^\star_1(D_0)$};
   	
   	\node[above=1cm of D1',align=center] {\small $\alpha_1$ is critical \\ \small w.r.t. $\mathcal{E}^\star_1(D_1)$};
	
    \path[->,
	line join=round, decoration={
	    zigzag,
	    segment length=4,
	    amplitude=.9,post=lineto,
	    post length=2pt}]	    
	    
	(C)	 edge[decorate] node[above, align=center] {$\beta$} (D0)
	(D0)	edge[decorate] node[above,align=center] {$\alpha_0$} (D0')
	(D0')	edge node[below left,align=center] {$p_{n-1}$} (pn-1)
	(pn-1)	edge node[above,align=center] {$c_{n-1}$} (D1)
	(D1) 	edge[decorate] node[above,align=center] {$\alpha_1$} node[below,align=center] {\small $p_{n-1}$-only} (D1')
	(D1')	 edge[decorate] node[below,align=center] {\small $p_{n-1}$-only} (D1'pn-1)
	(D0')	edge node[above left,align=center] {$p_iR_i$} (piRi)
	(piRi)	edge node[above,align=center] {$c_{n-1}$} (piRicn-1)
	(piRicn-1)	edge[decorate] node[below,align=center] {\small $p_{n-1}$-only} (piRipn-1);
\end{tikzpicture}
\bigskip
\caption{The construction of $D_1$ and $D_1'$ when $D_0'$ is neither $n$-recording nor $v$-hiding for any $v \in \{0, 1\}$. Any configuration $C$ with a double outline is bivalent with respect to $\mathcal{E}^\star_1(C)$.}\label{fig:d1}
\end{figure*}

\begin{theorem}\label{thm:main}
	If there exists a recoverable wait-free consensus algorithm for $n > 1$ processes using objects with deterministic types $\mathcal{T}_0, \mathcal{T}_1, \mathcal{T}_2, \ldots$, then at least one of the types $\mathcal{T}_i$ is $n$-recording.
\end{theorem}

\begin{proof}
We construct a chain of configurations $D_0, D_0', \ldots, D_\ell, D_\ell'$ where $\ell \leq n-1$ and a value $v \in \{0, 1\}$ such that every process is poised to access the same object $O$ in $D_0', \ldots, D_\ell'$, the configuration $D_\ell'$ is $n$-recording, and, for all $i \in \{0, \ldots, \ell\}$,
\begin{enumerate}[label=(\alph*)]
\item $D_i$ is bivalent with respect to $\mathcal{E}^\star_1(D_i)$,\label{prop:biv}
\item $D_i'$ is reachable from $D_i$ via a schedule $\alpha_i \in \mathcal{E}^\star_1(D_i)$,\label{prop:alpha}
\item $\alpha_i$ is critical with respect to $\mathcal{E}^\star_1(D_i)$,\label{prop:critical}
\item if $\ell > i > 0$, then $D_i'$ is $v$-hiding and not $n$-recording,\label{prop:vhiding}
\item $p_{n-i}, \ldots, p_{n-1}$ are on team $v$ in $D_i'$, and\label{prop:teamv}
\item $D_i'$ is reachable from $D_0$ via an execution $\delta_i$ containing only events by $p_{n-i}, \ldots, p_{n-1}$ such that, for each process $p_k \in \{p_{n-i}, \ldots, p_{n-1}\}$, the number of times $p_k$ crashes in $\delta_i$ is no more than $ns + i$, where $s$ is the number of steps by $p_{n-i}, \ldots, p_{k-1}$ in $\delta_i$.\label{prop:numcrash}
\end{enumerate}

\noindent Since $D_\ell'$ is $n$-recording and every process is poised to access $O$ in $D_\ell'$, this proves that the type of $O$ is $n$-recording.

\smallskip

We now show how to construct this sequence of configurations.
Let $C$ be some initial configuration in which $p_0$ has input $0$ and $p_1$ has input $1$.
By Observation~\ref{obs:bivinit}, $C$ is bivalent with respect to $\mathcal{E}^\star_1(C)$.
By Lemma~\ref{lem:eventualcrit}~\ref{lem:eventualcrit:b}, there exists a finite execution $\beta \in \mathcal{E}_2(C)$ such that $C\beta$ is bivalent with respect to $\mathcal{E}^\star_1(C\beta)$ and, for every nonempty finite schedule $\gamma$ such that $\beta\gamma \in \mathcal{E}_2(C)$, the configuration $C\beta\gamma$ is univalent with respect to $\mathcal{E}^\star_1(C\beta\gamma)$.

Define $D_0 = C\beta$.
Then $D_0$ is bivalent with respect to $\mathcal{E}_1^\star(D_0)$, giving us property~\ref{prop:biv} for $i = 0$.
Lemma~\ref{lem:eventualcrit}~\ref{lem:eventualcrit:a} implies that there exists a finite execution $\alpha_0 \in \mathcal{E}^\star_1(D_0)$ such that $\alpha_0$ is critical with respect to $\mathcal{E}^\star_1(D_0)$.
Define $D_0' = D_0\alpha_0$.
This gives us properties \ref{prop:alpha} and \ref{prop:critical} for $i = 0$.
By Lemma~\ref{lem:sameobject}, every process is poised to access the same object $O$ in $D_0\alpha_0$.
Properties \ref{prop:vhiding} and \ref{prop:teamv} hold vacuously for $i = 0$.

Since $\alpha_0 \in \mathcal{E}^\star_1(D_0)$ and $\mathcal{E}^\star_1(D_0) \subset \mathcal{E}_2(D_0)$, Observation~\ref{obs:closedconcat} implies that $\beta\alpha_0 \in \mathcal{E}_2(C)$.
Since $D_0' = C\beta\alpha_0$ is bivalent with respect to $\mathcal{E}^\star_1(D_0')$ by Lemma~\ref{lem:critisbiv}, $\alpha_0$ must be empty by Lemma~\ref{lem:eventualcrit}~\ref{lem:eventualcrit:b}.
Define $\delta_0 = \alpha_0$.
This gives us property~\ref{prop:numcrash} for $i = 0$.
If $D_0'$ is $n$-recording, then define $D_\ell = D_0 = C\beta$ and $D_\ell' = D_0' = C\beta\alpha$.
This completes the construction of $D_0$ and $D_0'$.

%


Before defining $D_i$ and $D_i'$, we first show how to define $D_1$ and $D_1'$ in the special case in which $D_0'$ is neither $n$-recording nor $v$-hiding for any $v \in \{0, 1\}$.




\subsubsection*{Defining $D_1$ and $D_1'$ when $D_0'$ is neither $n$-recording nor $v$-hiding}

By Observation~\ref{obs:cases}, there exists a process $p_i$ on team $v$ in $D_0'$, a process $p_j$ on team $\bar{v}$ in $D_0'$, an $R_i \in \mathcal{S}\bigl(\mathcal{P} - \{p_i\}\bigr)$, and an $R_j \in \mathcal{S}\bigl(\mathcal{P} - \{p_j\}\bigr)$ such that \emph{value}$(O, D_0' p_iR_i) =$ \emph{value}$(O, D_0' p_jR_j)$.
Without loss of generality, suppose that $p_{n-1}$ is on team $\bar{v}$.
By Lemma~\ref{lem:samevalrestrict}, $p_j = p_{n-1}$ and $R_j = \langle\rangle$.
Since $p_i$ is on team $v$ in $D_0'$ and $p_{n-1}$ is on team $\bar{v}$ in $D_0'$, we know that $p_i$ and $p_{n-1}$ are different processes, that is, $i < n-1$.
Since $p_i$ is on team $v$ in $D_0'$, it must be the case that $\alpha_0 p_iR_i$ is $v$-univalent with respect to $\mathcal{E}^\star_1(D_0)$.
We know that $p_iR_i c_{n-1} \in \mathcal{E}^\star_1(D_0')$ because $i < n-1$.
By Observation~\ref{obs:closedconcat}, $\alpha p_iR_i c_{n-1} \in \mathcal{E}^\star_1(D_0)$.
Therefore, Observation~\ref{obs:stilluniv} implies that $\alpha p_iR_ic_{n-1}$ is $v$-univalent with respect to $\mathcal{E}^\star_1(D_0)$.
Then $p_{n-1}$ decides $v$ in its solo-terminating execution from $D_0' p_iR_ic_{n-1}$.
Since all of the base objects have the same values in $D_0' p_iR_ic_{n-1}$ and $D_0' p_{n-1}c_{n-1}$ and these two configurations are indistinguishable to $p_{n-1}$, process $p_{n-1}$ decides $v$ in its solo-terminating execution from $D_0' p_{n-1}c_{n-1}$ as well.
Similarly, since all of the objects have the same values in $D_0' p_{n-1}$ and $D_0' p_{n-1}c_{n-1}$, these two configurations are indistinguishable to $p_0$, and $\alpha p_{n-1}$ is $\bar{v}$-univalent with respect to $\mathcal{E}^\star_1(D_0)$, we know that $p_0$ decides $\bar{v}$ in its solo-terminating execution from $D_0' p_{n-1}c_{n-1}$.
Therefore, $D_0' p_{n-1}c_{n-1}$ is bivalent with respect to $\mathcal{E}^\star_1(D_0' p_{n-1}c_{n-1})$.
Define $D_1 = D_0' p_{n-1}c_{n-1}$.
This gives us property~\ref{prop:biv} for $i = 1$.
By Lemma~\ref{lem:eventualcrit}~\ref{lem:eventualcrit:a}, there exists a finite execution $\alpha_1 \in \mathcal{E}^\star_1(D_1)$ such that $\alpha_1$ is critical with respect to $\mathcal{E}^\star_1(D_1)$.
Define $D_1' = D_1\alpha_1$.
This gives us properties \ref{prop:alpha} and \ref{prop:critical} for $i = 1$.

We claim that $\alpha_1$ contains no steps by $p_0, \ldots, p_{n-2}$.
To obtain a contradiction, suppose that $\alpha_1$ contains at least one step by some process $p_0, \ldots, p_{n-2}$.
Recall that $\alpha_0 \in \mathcal{E}^\star_1(D_0)$ and $\alpha_1 \in \mathcal{E}^\star_1(D_1)$.
Hence, if $s$ is the number of steps taken by $p_0, \ldots, p_{n-2}$ in $\alpha_0$ and $\alpha_1$ together, then $p_{n-1}$ crashes at most $ns$ times in these executions.
The processes $p_0, \ldots, p_{n-2}$ also take $s$ steps in $\alpha_0 p_{n-1}c_{n-1}\alpha_1$ and $p_{n-1}$ crashes at most $ns + 1$ times in this execution.
Since some process in $p_0, \ldots, p_{n-2}$ takes a step in $\alpha_1$, we have $s \geq 1$.
Furthermore, $n \geq 2$, and so $ns + 1 < 2ns$.
Then process $p_{n-1}$ crashes at most $2n$ times the number of steps by $p_0, \ldots, p_{n-2}$ in $\alpha_0 p_{n-1}c_{n-1}\alpha_1$.
Hence, $\alpha_0 p_{n-1}c_{n-1}\alpha_1 \in \mathcal{E}_2(D_0)$.
Recall that $D_0 = C\beta$, where $\beta \in \mathcal{E}_2(C)$.
Observation~\ref{obs:closedconcat}  implies that $\beta\alpha_0 p_{n-1}c_{n-1}\alpha_1 \in \mathcal{E}_2(C)$.
Furthermore, since $\alpha_1$ is critical with respect to $\mathcal{E}^\star_1(D_1)$, the configuration $D_1' = C\beta\alpha_0 p_{n-1}c_{n-1}\alpha_1$ is bivalent with respect to $\mathcal{E}^\star_1(D_1')$ by Lemma~\ref{lem:critisbiv}.
This contradicts Lemma~\ref{lem:eventualcrit}~\ref{lem:eventualcrit:b}.
Therefore, $\alpha_1$ only contains steps by $p_{n-1}$.

Since $\alpha_1$ only contains events by $p_{n-1}$ and $p_{n-1}$ decides $v$ in its solo-terminating execution from $D_1$, we know that $p_{n-1}$ decides $v$ in its solo-terminating execution from $D_1' = D_1\alpha_1$ as well.
Hence, $p_{n-1}$ is on team $v$ in $D_1'$.
This gives us property~\ref{prop:teamv} for $i = 1$.

Define $\delta_1 = \alpha_0p_{n-1}c_{n-1}\alpha_1$.
Since $\alpha_1 \in \mathcal{E}^\star_1(D_1)$ and $\alpha_1$ contains only steps by $p_{n-1}$, no crashes occur in $\alpha_1$.
Since $\alpha_0$ is empty, process $p_{n-1}$ crashes exactly once in $\delta_1 = \alpha_0p_{n-1}c_{n-1}\alpha_1$.
This gives us property~\ref{prop:numcrash} for $i = 1$.

By Lemma~\ref{lem:sameobject}, every process is poised to access the same object in $D_1'$.
Since $p_0, \ldots, p_{n-1}$ are poised to access $O$ in $D_0'$ and $p_0, \ldots, p_{n-2}$ take no steps during $p_{n-1}c_{n-1}\alpha'$, every process is poised to access $O$ in $D_1'$ as well.
By Lemma~\ref{lem:maintech}, $D_1'$ is either $n$-recording or $v$-hiding.
If $D_1'$ is $n$-recording, then define $D_\ell = D_1$ and $D_\ell' = D_1'$.
Otherwise, $D_1'$ is $v$-hiding and not $n$-recording, which gives us property~\ref{prop:vhiding} for $i = 1$.
This case is depicted in Figure~\ref{fig:d1}.

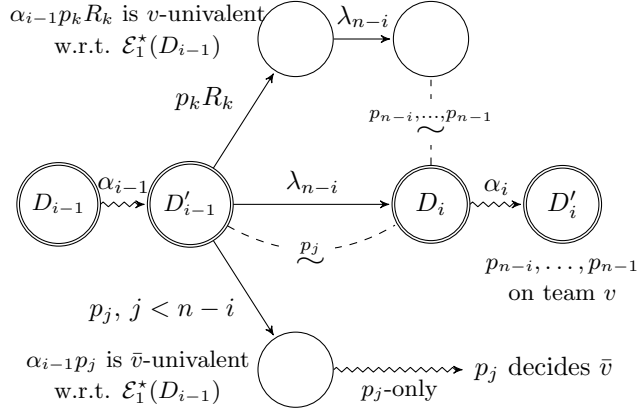
\begin{figure}[t]
\centering
\begin{tikzpicture}[shorten >=1pt,node distance=2.9cm,on grid,auto] 
	\tikzset{every state/.style={minimum size=1cm}}

	\node[state,accepting]		(Di-1) 	{\small $D_{i-1}$};
	\node[state,accepting]		(Di-1') [right=1.75cm of Di-1]		{\small $D_{i-1}'$};

	\node[state]		(pkRk)	[above right=2.2cm and 1.4cm of Di-1'] {};
	\node[state]		(pkRklambda) [right=1.8cm of pkRk] {};
	
	\node[state,accepting]		(Di) [right=3.2cm of Di-1']		{$D_i$};
	\node[state,accepting]		(Di') [right=1.75cm of Di] {$D_i'$};
	
	\node[state]					(pji-1) [below right=2.2cm and 1.4cm of Di-1']	{};
	\node				(pji-1solo) [right=3.3cm of pji-1]	{$p_{j}$ decides $\bar{v}$};

	
%
%
   	\node[above left=0.1cm and 2.1cm of pkRk,align=center] {\small $\alpha_{i-1}p_kR_k$ is $v$-univalent \\ \small w.r.t. $\mathcal{E}^\star_1(D_{i-1})$};
   	
   	\node[below left=0.1cm and 2.1cm of pji-1,align=center] {\small $\alpha_{i-1}p_{j}$ is $\bar{v}$-univalent \\ \small w.r.t. $\mathcal{E}^\star_1(D_{i-1})$};
   	
   	\node[below=1cm of Di',align=center] {\small $p_{n-i}, \ldots, p_{n-1}$ \\ \small on team $v$};
   	
   	\draw [-,dashed] (Di-1') to [bend right=30] (Di);
   	\node[below right=0.65cm and 1.6cm of Di-1',fill=white] {$\widesim{p_{j}}$};
   	
   	\draw [-,dashed] (Di) to (pkRklambda);
   	\node[above=1.1cm of Di,fill=white] {$\widesim{p_{n-i}, \ldots, p_{n-1}}$};

%
	
    \path[->,
	line join=round, decoration={
	    zigzag,
	    segment length=4,
	    amplitude=.9,post=lineto,
	    post length=2pt}]	    
	    
	(Di-1)	edge[decorate] node[above,align=center] {$\alpha_{i-1}$} (Di-1')
	(Di-1') 	edge node[above left] {$p_kR_k$} (pkRk)
	(pkRk)	 edge node[above,align=center] {$\lambda_{n-i}$} (pkRklambda)
	(Di-1')	edge node[above,align=center] {$\lambda_{n-i}$} (Di)
	(Di)	edge[decorate] node[above,align=center] {$\alpha_i$} (Di')
	(Di-1') edge node[below left] {$p_{j}$, $j < n-i$} (pji-1)
	(pji-1) edge[decorate] node[below] {\small $p_{j}$-only}  (pji-1solo);
\end{tikzpicture}
\bigskip
\caption{The construction of $D_i$ and $D_i'$. Configurations $C$ with a double outline are bivalent w.r.t. $\mathcal{E}^\star_1(C)$.}\label{fig:di}
\end{figure}


\subsubsection*{Defining $D_i$ and $D_i'$}

Suppose that we have constructed the sequence $D_0, D_0', \ldots, D_{i-1}, D_{i-1}'$ for some $i \geq 1$.
(Note that in the previous subsection, we covered the case in which $D_0'$ is not $v$-hiding, for any $v \in \{0, 1\}$; therefore, we only consider the case in which $D_{i-1}'$ is $v$-hiding.)
Define $D_i = D_{i-1}'\lambda_{n-i}$.


By property~\ref{prop:teamv} for $i - 1$, all of the processes $p_{n-i+1}, \ldots, p_{n-1}$ are on team $v$ in $D_{i-1}'$.
Furthermore, if $p_{n-i}$ is the only process on team $\bar{v}$ in $D_{i-1}'$, then $D_{i-1}'$ is $n$-recording, which contradicts property~\ref{prop:vhiding} for $i-1$.
Hence, there is some process $p_j$ on team $\bar{v}$ in $D_{i-1}'$ such that $j < n-i$.

Since $D_{i-1}'$ is $v$-hiding by property~\ref{prop:vhiding} for $i-1$, there exists a process $p_k$ on team $v$ and an $R_k \in \mathcal{S}\bigl(\mathcal{P} - \{p_k\}\bigr)$ such that $\textit{value}(O, D_{i-1}'p_kR_k) = \textit{value}(O, D_{i-1}')$.
We claim that $p_{j} \in R_k$.
To obtain a contradiction suppose that $p_{j} \not\in R_k$.
Then $D_{i-1}'p_{j}$ and $D_{i-1}'p_kR_kp_{j}$ are indistinguishable to $p_{j}$.
Furthermore, since all of the objects have the same values in $D_{i-1}'$ and $D_{i-1}'p_kR_k$, all of the objects have the same values in $D_{i-1}'p_{j}$ and $D_{i-1}'p_kR_kp_{j}$ as well.
Therefore, $p_{j}$ decides the same value in its solo-terminating executions from $D_{i-1}'p_{j}$ and $D_{i-1}'p_kR_kp_{j}$.
However, since $p_k$ is on team $v$ in $D_{i-1}'$ and $p_{j}$ is on team $\bar{v}$ in $D_{i-1}'$, the execution $\alpha_{i-1} p_kR_kp_{j}$ is $v$-univalent with respect to $\mathcal{E}^\star_1(D_{i-1})$ and $\alpha_{i-1} p_{j}$ is $\bar{v}$-univalent with respect to $\mathcal{E}^\star_1(D_{i-1})$.
Therefore, $p_{j}$ decides $v$ in its solo-terminating execution from $D_{i-1}'p_kR_kp_{j}$ and $\bar{v}$ in its solo-terminating execution from $D_{i-1}'p_{j}$.
This is a contradiction.
Hence, $p_{j} \in R_k$.

Since $p_{j} \in R_k$ and $j < n-i$, we know that $\textit{exec}(D_{i-1}', p_kR_k\lambda_{n-i}) \in \mathcal{E}^\star_1(D_{i-1}')$.
By Observation~\ref{obs:closedconcat} and the fact that $\alpha_{i-1} \in \mathcal{E}^\star_1(D_{i-1})$, it follows that $\alpha_{i-1}p_kR_k\lambda_{n-i} \in \mathcal{E}^\star_1(D_{i-1})$.
Since $\alpha_{i-1}p_kR_k$ is $v$-univalent with respect to $\mathcal{E}^\star_1(D_{i-1})$, Observation~\ref{obs:stilluniv} implies that $\alpha_{i-1}p_kR_k\lambda_{n-i}$ is $v$-univalent with respect to $\mathcal{E}^\star_1(D_{i-1})$ as well.
Since $D_{i-1}\alpha_{i-1}p_kR_k\lambda_{n-i}$ and $D_i = D_{i-1}'\lambda_{n-i}$ are indistinguishable to $p_{n-i}, \ldots, p_{n-1}$ and all of the objects have the same values in these two configurations, Observation~\ref{obs:indistuniv} implies that $\{p_{n-i}, \ldots, p_{n-1}\}$ is $v$-univalent in $D_i$ with respect to $\mathcal{E}^\star_1(D_i)$ as well.


Since $p_{j}$ is on team $\bar{v}$ in $D_{i-1}'$, it decides $\bar{v}$ in its solo-terminating execution from $D_{i-1}'$.
Notice that $D_{i-1}'$ and $D_i = D_{i-1}'\lambda_{n-i}$ are indistinguishable to $p_{j}$ and all of the objects have the same values in these two configurations.
Therefore, $p_{j}$ decides $\bar{v}$ in its solo-terminating execution from $D_i$ as well.
Hence, $D_i$ is bivalent with respect to $\mathcal{E}^\star_1(D_i)$.
This gives us property~\ref{prop:biv} for $i$.
By Lemma~\ref{lem:eventualcrit}~\ref{lem:eventualcrit:a}, there exists a finite schedule $\alpha_i$ such that $\alpha_i \in \mathcal{E}^\star_1(D_i)$ and $\alpha_i$ is critical with respect to $\mathcal{E}^\star_1(D_i)$.
Define $D_i' = D_i\alpha_i$.
This gives us properties \ref{prop:alpha} and \ref{prop:critical} for $i$.
The construction of $D_i$ and $D_i'$ is depicted in Figure~\ref{fig:di}.

We claim that the execution $\alpha_i$ contains no steps by $p_0, \ldots, p_{n-i-1}$.
To obtain a contradiction, suppose the claim is false.
By property~\ref{prop:numcrash} for $i-1$, $D_{i-1}'$ is reachable from $D_0$ via an execution $\delta_{i-1}$ containing only events by $p_{n-i+1}, \ldots, p_{n-1}$ such that, for every process $p_k \in \{p_{n-i+1}, \ldots, p_{n-1}\}$, the number of times $p_k$ crashes in $\delta_{i-1}$ is no more than $ns_k + i - 1$, where $s_k$ is the number of steps by $p_{n-i}, \ldots, p_{k-1}$ in $\delta_{i-1}$.
Since $\alpha_{i} \in \mathcal{E}^\star_1(D_i)$, the number of times any process $p_k$ crashes in $\alpha_i$ is no more than $ns_k'$, where $s_k'$ is the number of steps by $p_{0}, \ldots, p_{k-1}$ in $\alpha_i$.
Since each process $p_{n-i}, \ldots, p_{n-1}$ crashes exactly once in $\lambda_{n-i}$, every process $p_k \in \{p_{n-i}, \ldots, p_{n-1}\}$ crashes at most $n(s_k + s_k') + i$ times in $\delta_{i-1}\lambda_{n-i}\alpha_i$.
Since $i \leq n-1$ and $s_k' \geq 1$ for all $p_k \in \{p_{n-i}, \ldots, p_{n-1}\}$ by assumption, we have $n(s_k + s_k') + i \leq 2n(s_k + s_k')$.
In other words, the number of crashes by each process $p_k \in \{p_{n-i}, \ldots, p_{n-1}\}$ in $\delta_{i-1}\lambda_{n-i}\alpha_i$ is no more than $2n$ times the number of steps taken by $p_0, \ldots, p_{k-1}$ in $\delta_{i-1}\lambda_{n-i}\alpha_i$.
Since $\delta_{i-1}\lambda_{n-i}$ contains no events by $p_0, \ldots, p_{n-i-1}$ and $\alpha_i \in \mathcal{E}^\star_1(D_i)$, the number of crashes by each process $p_k \in \{p_0, \ldots, p_{n-i-1}\}$ in $\delta_{i-1}\lambda_{n-i}\alpha_i$ is no more than $n$ times the number of steps taken by $p_0, \ldots, p_{k-1}$ in $\delta_{i-1}\lambda_{n-i}\alpha_i$.
Hence, $\delta_{i-1}\lambda_{n-i}\alpha_i \in \mathcal{E}_2(D_0)$.
Observation~\ref{obs:closedconcat} implies that $\beta\delta_{i-1}\lambda_{n-i}\alpha_i \in \mathcal{E}_2(C)$.
However, Lemma~\ref{lem:critisbiv} implies that $D_i' = C\beta\delta_{i-1}\lambda_{n-i}\alpha_i$ is bivalent with respect to $\mathcal{E}^\star_1(D_i')$.
This contradicts Lemma~\ref{lem:eventualcrit}~\ref{lem:eventualcrit:b}.
Hence, $\alpha_i$ contains no steps by $p_0, \ldots, p_{n-i-1}$.
By definition of $\mathcal{E}^\star_1(D_i)$, $\alpha_i$ also contains no crashes by $p_0, \ldots, p_{n-i-1}$.
Therefore, $\alpha_i$ contains no events by $p_0, \ldots, p_{n-i-1}$.

Define $\delta_i = \delta_{i-1}\lambda_{n-i}\alpha_i$.
Since $\delta_i$ has no events by $p_0, \ldots, p_{n-i-1}$, each process $p_k \in \{p_{n-i}, \ldots, p_{n-1}\}$ crashes at most $n(s_k + s_k') + i$ times in $\delta_i$, where $s_k$ is the number of steps taken by $p_{n-i}, \ldots, p_{n-1}$ in $\delta_{i-1}$ and $s_k'$ is the number of steps taken by $p_{n-i}, \ldots, p_{n-1}$ in $\alpha_i$.
This gives us property~\ref{prop:numcrash} for $i$.

By Lemma~\ref{lem:maintech}, $D_i'$ is either $n$-recording or $v$-hiding.
If $D_i'$ is $n$-recording, then define $D_\ell = D_i$ and $D_\ell' = D_i'$.
Otherwise, $D_i'$ is $v$-hiding and not $n$-recording, which gives us property~\ref{prop:vhiding} for $i$.
Since $\{p_{n-i}, \ldots, p_{n-1}\}$ is $v$-univalent in $D_i = D_{i-1}'\lambda_{n-i}$ with respect to $\mathcal{E}^\star_1(D_i)$ and $\alpha_i$ only contains events by $p_{n-i}, \ldots, p_{n-1}$, Observation~\ref{obs:stilluniv} implies that $\{p_{n-i}, \ldots, p_{n-1}\}$ is $v$-univalent in $\alpha_i$ with respect to $\mathcal{E}^\star_1(D_i)$.
Therefore, $p_{n-i}, \ldots, p_{n-1}$ are all on team $v$ in $D_i'$.
This gives us property~\ref{prop:teamv} for $i$.
This completes the construction.

%
%
\end{proof}

Delporte-Gallet, Fatourou, Fauconnier, and Ruppert \cite{dffr-22} showed that, if a deterministic, readable type is $n$-recording, then objects of that type can be used along with registers to solve recoverable wait-free consensus among $n$ processes.
Their result combined with Theorem~\ref{thm:main} gives us the following.


\begin{theorem}
The recoverable consensus hierarchy is robust for deterministic, readable types.
\end{theorem}

%

\section{Non-readable Types}\label{sec:nonreadable}

For all $n > n' \geq 1$, we describe a non-readable type $\mathcal{T}_{n, n'}$ that has consensus number $n$ and recoverable consensus number $n'$.
The type $\mathcal{T}_{n, n'}$ has three operations, $op_0$, $op_1$, and $op_R$, and $2n$ values $s$, $s_\bot$, and $s_{x, i}$, where $x \in \{0, 1\}$ and $i \in \{1, \ldots, n-1\}$.
Applying $op_0$ to an object of $\mathcal{T}_{n, n'}$ with value $s$ returns $0$ and changes its value to $s_{0, 1}$.
Similarly, applying $op_1$ to an object of $\mathcal{T}_{n, n'}$ with value $s$ returns $1$ and changes its value to $s_{1, 1}$.
Applying either $op_0$ or $op_1$ to an object with value $s_{x, i}$, for some $x \in \{0, 1\}$ and $i < n - 1$, returns $x$ and changes its value to $s_{x, i+1}$.
Applying either $op_0$ or $op_1$ to an object with value $s_{x, n-1}$ returns $x$ and changes the value to $s_\bot$.
When the object has value $s_\bot$, applying any operation returns $\bot$ and does not change the value of the object.

The operation $op_R$ is similar to \emph{Read} unless the object has value $s_{x, i}$ for some $x \in \{0, 1\}$ and $i > n'$.
More specifically, when an object of $\mathcal{T}_{n, n'}$ has value $s$, applying $op_R$ returns $s$ and does not change the value.
Applying $op_R$ when the object has value $s_{x, i}$ where $i \leq n'$ returns $s_{x, i}$ and does not change the value.
If $i > n'$, then applying $op_R$ to an object with value $s_{x, i}$ returns $\bot$ and changes its value to $s_\bot$.
\ifarxiv
The state machine diagram of type $\mathcal{T}_{5, 2}$ is given in Figure~\ref{fig:nonreadabletype}.

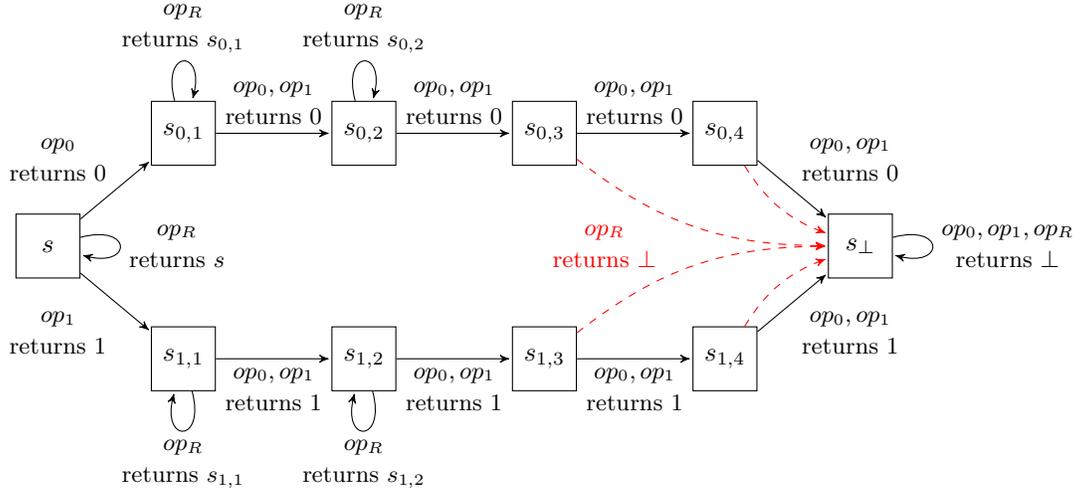
\begin{figure}[h]
\centering
\begin{tikzpicture}[shorten >=1pt,node distance=2.9cm,on grid,auto] 
	\tikzset{state/.style={rectangle,draw,minimum width=0.85cm,minimum height=0.85cm}}

	\node[state]		(s)		{$s$};
	\node[state]		(sA1) [above right=1.5cm and 1.8cm of s]	{$s_{0,1}$};
	\node[state]		(sA2) [right=2.4cm of sA1] 	{$s_{0,2}$};
	\node[state]		(sA3) [right=2.4cm of sA2]		{$s_{0,3}$};
	\node[state]		(sA4) [right=2.4cm of sA3]		{$s_{0,4}$};
	
	\node	(rdlabel)	[below right=1.5cm and 0.8cm of sA3,align=center,red]	{\small $op_R$ \\ \small returns $\bot$};
	
	\node[state]		(sB1) [below right=1.5cm and 1.8cm of s]	{$s_{1,1}$};
	\node[state]		(sB2) [right=2.4cm of sB1]	{$s_{1,2}$};
	\node[state]		(sB3) [right=2.4cm of sB2]	{$s_{1,3}$};
	\node[state]		(sB4) [right=2.4cm of sB3]	{$s_{1,4}$};
	
	\node[state]		(sbot) [below right=1.5cm and 1.8cm of sA4]		{$s_{\bot}$};

    \path[->,
	line join=round, decoration={
	    zigzag,
	    segment length=4,
	    amplitude=.9,post=lineto,
	    post length=2pt}]	    
	    
	(s)	 edge node[above left,align=center] {\small $op_0$ \\ \small returns $0$} (sA1)
	(sA1)	edge node[above,align=center] {\small $op_0, op_1$ \\ \small returns $0$} (sA2)
	(sA2)	edge node[above,align=center] {\small $op_0, op_1$ \\ \small returns $0$} (sA3)
	(sA3)	edge node[above,align=center] {\small $op_0, op_1$ \\ \small returns $0$} (sA4)
	(sA4)	edge node[above right,align=center] {\small $op_0, op_1$ \\ \small returns $0$} (sbot)
	
	(s)	 edge node[below left,align=center] {\small $op_1$ \\ \small returns $1$} (sB1)
	(sB1)	edge node[below,align=center] {\small $op_0, op_1$ \\ \small returns $1$} (sB2)
	(sB2)	edge node[below,align=center] {\small $op_0, op_1$ \\ \small returns $1$} (sB3)
	(sB3)	edge node[below,align=center] {\small $op_0, op_1$ \\ \small returns $1$} (sB4)
	(sB4)	edge node[below right,align=center] {\small $op_0, op_1$ \\ \small returns $1$} (sbot)
	
	(sA3)	edge[bend right=19,red,dashed]	(sbot)
	(sA4)	edge[bend right=19,red,dashed]	(sbot)
	(sB3)	edge[bend left=19,red,dashed]	(sbot)
	(sB4)	edge[bend left=19,red,dashed]	(sbot)
	
	(sA1)	edge[loop above] node[above,align=center] {\small $op_R$ \\ \small returns $s_{0,1}$} (sA1)
	(sA2)	edge[loop above] node[above,align=center] {\small $op_R$ \\ \small returns $s_{0,2}$} (sA2)
	(sB1)	edge[loop below] node[below,align=center] {\small $op_R$ \\ \small returns $s_{1,1}$} (sB1)
	(sB2)	edge[loop below] node[below,align=center] {\small $op_R$ \\ \small returns $s_{1,2}$} (sB2)
	(sbot)	edge[loop right] node[right,align=center] {\small $op_0, op_1, op_R$ \\ \small returns $\bot$} (sbot)
	(s)	edge[loop right] node[right,align=center] {\small $op_R$ \\ \small returns $s$} (s);
\end{tikzpicture}
\medskip
\caption{The state machine diagram of the type $\mathcal{T}_{5, 2}$.}\label{fig:nonreadabletype}
\end{figure}
\fi

\smallskip

We now argue that wait-free binary consensus can be solved among $n$ processes using a single object $O$ of type $\mathcal{T}_{n, n'}$.
The object $O$ begins with value $s$.
A process with input $x \in \{0, 1\}$ applies $op_x$ to $O$ and decides the value returned by the operation.
Notice that the first operation applied to $O$ in the algorithm determines the value returned by the next $n-1$ operations applied to $O$.
Therefore, the algorithm satisfies agreement.
The object $O$ can only have value $s_{v, i}$, for some $v \in \{0, 1\}$ and some $i \in \{1, \ldots, n-1\}$, if the first operation applied to $O$ was $op_v$.
Since a process only applies $op_v$ when it has input $v$, the algorithm satisfies validity.
Finally, each process applies only one operation, so the algorithm satisfies wait-freedom.
A standard valency argument shows that consensus number of $\mathcal{T}_{n, n'}$ is at most $n$.
\ifarxiv
We perform this valency argument in the proof of the following lemma.
\fi

\begin{lemma}\label{lem:consn}
	The consensus number of type $\mathcal{T}_{n, n'}$ is $n$.
\end{lemma}

\ifarxiv
\begin{proof}
We already showed that the consensus number of $\mathcal{T}_{n, n'}$ is at least $n$.
To complete the proof, we now show that the consensus number of $\mathcal{T}_{n, n'}$ is at most $n$.
For this proof, we use standard definitions of valency.
A configuration $C$ is \emph{bivalent} if, for all $v \in \{0, 1\}$, there exists an execution $\beta_v$ from $C$ such that some process has decided $v$ in $C\beta_v$.
A configuration is \emph{$v$-univalent} if, for every execution $\beta$ from $C$ such that some process has decided a value in $C\beta$, that value is $v$.
To obtain a contradiction, suppose that there is a wait-free algorithm that solves consensus among $n+1$ processes $p_0, \ldots, p_{n}$ using instances of $\mathcal{T}_{n, n'}$ along with registers.
Let $C$ be a bivalent initial configuration of the algorithm in which process $p_0$ has input $0$ and $p_1$ has input $1$.
Using a standard argument, we can show that there is an execution from $C$ that leads to a \emph{critical} configuration $C'$, that is, the configuration $C'$ is bivalent and $C'p_i$ is univalent, for all $i$.
Since $n+1 > 1$, a standard argument shows that every process is poised to access the same object $O$ of type $\mathcal{T}_{n, n'}$ in $C'$.
Define $T_v$ to be the set of processes $p_i$ such that $C'p_i$ is $v$-univalent, and define $o_i$ as the operation that $p_i$ is poised to apply to $O$ in $C'$ for all $i \in \{0, \ldots, n\}$.

First suppose that $o_i = op_R$ for some $i$, where $p_i \in T_v$.
Then either $\textit{value}(O, C'p_i) = \textit{value}(O, C')$ or $\textit{value}(O, C'p_i) = s_\bot$.
In either case, $\textit{value}(O, C'p_ip_j) = \textit{value}(O, C'p_j)$, where $p_j \in T_{\bar{v}}$.
Since $n > n' \geq 1$, we have $n+1 \geq 3$.
Hence, there is a process $p_k$ different from $p_i$ and $p_j$.
Then $C'p_ip_j$ and $C'p_j$ are indistinguishable to $p_k$ and all of the objects have the same values in these two configurations.
Hence, $p_k$ decides the same value in its solo-terminating executions from $C'p_ip_j$ and $C'p_j$.
However, $C'p_ip_j$ is $v$-univalent and $C'p_j$ is $\bar{v}$-univalent.
This is a contradiction
Hence, all of the operations $o_0, \ldots, o_n$ are either $op_0$ or $op_1$.

Let $p_i \in T_0$, $p_j \in T_1$, and let $p_k$ be some process other than $p_i$ and $p_j$.
Since each operation $o_0, \ldots o_n$ is either $op_0$ or $op_1$, for any schedule $\sigma$ containing all $n$ processes besides $p_k$, $\textit{value}(O, C'\sigma) = s_\bot$.
Let $\sigma_i$ and $\sigma_j$ be two schedules containing all $n$ processes besides $p_k$ such that $\sigma_i$ begins with $p_i$ and $\sigma_j$ begins with $p_j$.
Then $C'\sigma_i$ and $C'\sigma_j$ are indistinguishable to $p_k$ and all of the objects have the same values in these two configurations.
Hence, $p_k$ decides the same value in its solo-terminating executions from $C'\sigma_i$ and $C'\sigma_j$.
However, $C'\sigma_i$ is $v$-univalent and $C'\sigma_j$ is $\bar{v}$-univalent.
This is a contradiction.
Hence, the consensus number of $\mathcal{T}_{n, n'}$ is at most $n$.
This completes the proof that the consensus number of $\mathcal{T}_{n, n'}$ is $n$.
\end{proof}
\else
\fi

We first argue that recoverable wait-free binary consensus can be solved among $n'$ processes using a single object $O$ of type $\mathcal{T}_{n, n'}$.
The object $O$ begins with value $s$.
A process with input $x \in \{0, 1\}$ first applies $op_R$.
If the operation returns a value $s_{v,i}$, then the process decides $v$.
If the operation returns $\bot$, then the process decides $0$ (we will argue that this never happens).
Otherwise, the operation returns the initial value $s$.
In this case, the process applies $op_x$ and then decides the value returned.

Notice that any process with input $x \in \{0, 1\}$ applies at most one $op_x$ operation, because $op_0$ and $op_1$ change the value of $O$ from $s$ to $s_{0, 1}$ and $s_{1, 1}$, respectively, and once $O$ has a value different from $s$ it can never have value $s$ again.
To change the value of $O$ from its initial value $s$ to $s_{v, i}$, at least $i$ $op_0$ and $op_1$ operations must be applied.
Since there are $n'$ processes, the object $O$ can never have value $s_{v, i}$ where $i > n'$ during any execution of our algorithm.
Therefore, none of the $op_R$ operations can return $\bot$, and the first $op_x$ operation applied to $O$ where $x \in \{0, 1\}$ determines the value returned by all subsequent operations.
Hence, the algorithm satisfies agreement.
The object $O$ can only have value $s_{v, i}$, for some $v \in \{0, 1\}$ and some $i \in \{1, \ldots, n'\}$, if the first operation applied to $O$ was $op_v$.
A process only applies $op_v$ if it has input $v$, so the algorithm satisfies validity.
Finally, starting from its initial state, each process applies at most $2$ operations to $O$ before deciding, so the algorithm satisfies recoverable wait-freedom.

By considering the set of executions $\mathcal{E}^\star_m(C)$, where $m = \lceil\frac{n}{2}\rceil$ and $C$ is an initial configuration that is bivalent with respect to $\mathcal{E}^\star_m(C)$, we can perform a valency argument to prove that the recoverable consensus number of $\mathcal{T}_{n, n'}$ is at most $n'$.
First, we use Lemma~\ref{lem:eventualcrit}~\ref{lem:eventualcrit:a} to obtain an execution $\alpha$ that is critical with respect to $\mathcal{E}^\star_m(C)$.
Lemma~\ref{lem:sameobject} can be used along with a standard argument to show that every process is poised to access the same object $O$ of type $\mathcal{T}_{n, n'}$ in $C\alpha$.
We then show that every process must be poised to apply $op_0$ or $op_1$ in $C\alpha$.
Consider two permutations $\sigma_0$ and $\sigma_1$ of $p_0, \ldots, p_{n'}$, where $\sigma_0$ begins with a process on team $0$ in $C\alpha$ and $\sigma_1$ begins with a process on team $1$ in $C\alpha$.
Since each process applies either $op_0$ or $op_1$ in $\sigma_0$ and $\sigma_1$, we argue that the value of $O$ in $C\alpha\sigma_0$ and $C\alpha\sigma_1$ must be either $s_\bot$ or $s_{v, i}$ for some $v \in \{0, 1\}$ and $i > n'$.
Hence, in any executions from $C\alpha\sigma_0$ or $C\alpha\sigma_1$, any process that applies $op_R$ to $O$ will ``break'' the object by changing its value to $s_\bot$.
Using this, we construct two executions $\rho$ and $\tau$ from $C\alpha\sigma_0$ and $C\alpha\sigma_1$, respectively, each containing only steps by $p_{n'}$ such that $p_{n'}$ decides the same value in $\rho$ and $\tau$.
By showing that $\alpha\sigma_0\rho$ and $\alpha\sigma_1\tau$ are both in $\mathcal{E}^\star_m(C)$, we obtain a contradiction.
\ifarxiv
We formalize this valency argument in the proof of the following lemma.
\fi

\begin{lemma}\label{lem:rconsn'}
The recoverable consensus number of type $\mathcal{T}_{n, n'}$ $n'$.
\end{lemma}

\ifarxiv
\begin{proof}
We already showed that the recoverable consensus number of $\mathcal{T}_{n, n'}$ is at least $n'$.
We now show that the recoverable consensus number of $\mathcal{T}_{n, n'}$ is at most $n'$.
To obtain a contradiction, suppose that there is a recoverable wait-free consensus algorithm for $n'+1$ processes $p_0, \ldots, p_{n'}$ using objects of $\mathcal{T}_{n, n'}$ and registers.
Let $C$ be an initial configuration of this algorithm in which $p_0$ has input $0$ and $p_1$ has input $1$.
Let $m = \lceil\frac{n}{2}\rceil$.
By Observation~\ref{obs:bivinit}, $C$ is bivalent with respect to $\mathcal{E}^\star_m(C)$.
By Lemma~\ref{lem:eventualcrit}~\ref{lem:eventualcrit:a}, there exists a finite execution $\alpha \in \mathcal{E}^\star_m(C)$ such that $\alpha$ is critical with respect to $\mathcal{E}^\star_m(C)$.
Lemma~\ref{lem:diffteams} implies that there is a process on team $0$ in $C\alpha$ and there is a process on team $1$ in $C\alpha$.

Lemma~\ref{lem:sameobject} implies that all of the processes are poised to access the same object $O$ in $C\alpha$.
Since there are $n'+1$ processes and $n'+1 > 1$, a standard argument can be used to show that $O$ is not a register, that is, $O$ is an object of $\mathcal{T}_{n, n'}$.
For all $i \in \{0, \ldots, n'\}$, define $o_i$ as the operation that $p_i$ is poised to apply to $O$ in $C\alpha$.
First suppose that $o_i = op_R$ for some $i$, where $p_i$ is on team $v$ in $C\alpha$.
Then either $\textit{value}(O, C\alpha p_i) = \textit{value}(O, C\alpha)$ or $\textit{value}(O, C\alpha p_i) = s_\bot$.
Let $p_j$ be a process on team $\bar{v}$ in $C\alpha$.
In case $\textit{value}(O, C\alpha p_i) = \textit{value}(O, C\alpha)$, notice that $C\alpha p_i$ and $C\alpha$ are indistinguishable to $p_j$ and all of the objects have the same values in these two configurations, so $p_j$ decides the same value in its solo-terminating executions from $C\alpha p_ip_j$ and $C\alpha p_j$.
However, $\alpha p_ip_j$ is $v$-univalent with respect to $\mathcal{E}^\star_m(C)$ and $\alpha p_j$ is $\bar{v}$-univalent with respect to $\mathcal{E}^\star_m(C)$, which is a contradiction.
Otherwise, $\textit{value}(O, C\alpha p_i) = s_\bot$.
Notice that $\textit{value}(O, C\alpha p_j p_i) = s_\bot$ in this case, no matter what $o_j$ is.
Suppose that $i < j$ (the case in which $j < i$ is symmetric), so that $\textit{exec}(C\alpha, p_ic_j)$ and $\textit{exec}(C\alpha, p_j p_ic_j)$ are in $\mathcal{E}^\star_m(C\alpha)$.
By Observation~\ref{obs:closedconcat}, $\alpha p_ic_j$ and $\alpha p_jp_ic_j$ are in $\mathcal{E}^\star_m(C)$.
Hence, by Observation~\ref{obs:stilluniv}, $\alpha p_ic_j$ is $v$-univalent with respect to $\mathcal{E}^\star_m(C)$ and $\alpha p_jp_ic_j$ is $\bar{v}$-univalent with respect to $\mathcal{E}^\star_m(C)$.
But $C\alpha p_ic_j$ and $C\alpha p_jp_ic_j$ are indistinguishable to $p_j$ and all of the objects have the same values in these two configurations.
Therefore, $p_j$ decides the same values in its solo-terminating executions from $C\alpha p_ic_j$ and $C\alpha p_jp_ic_j$, which is a contradiction.
Hence, $o_0, \ldots, o_{n'}$ are all either $op_0$ or $op_1$.

Let $\sigma_0$ be some permutation of $p_0, \ldots, p_{n'}$ beginning with a process on team $0$ in $C\alpha$, and let $\sigma_1$ be a permutation beginning with a process on team $1$ in $C\alpha$.
We proceed by constructing an execution $\rho$ from $C\alpha\sigma_0$ and an execution $\tau$ from $C\alpha\sigma_1$ such that $\rho$ and $\tau$ contain only steps by $p_{n'}$, the executions $\textit{exec}(C\alpha, \sigma_0)\rho$ and $\textit{exec}(C\alpha, \sigma_1)\tau$ are both in $\mathcal{E}^\star_m(C\alpha)$, $C\alpha\sigma_0\rho \widesim{p_{n'}} C\alpha\sigma_1\tau$, and $p_{n'}$ has decided the same value in $C\alpha\sigma_0\rho$ and $C\alpha\sigma_1\tau$.
Observation~\ref{obs:closedconcat} implies that $\alpha\sigma_0\rho \in \mathcal{E}^\star_m(C)$ and $\alpha\sigma_1\tau \in \mathcal{E}^\star_m(C)$.
Hence, by Observation~\ref{obs:stilluniv}, $\alpha\sigma_0\rho$ is $0$-univalent with respect to $\mathcal{E}^\star_m(C)$ and $\alpha\sigma_1\tau$ is $1$-univalent with respect to $\mathcal{E}^\star_m(C)$.
This is a contradiction.

We will construct $\rho$ and $\tau$ inductively.
We define $\rho_0, \rho_1, \ldots, \rho_{\ell}$ and $\tau_0, \tau_1, \ldots, \tau_{\ell}$ such that $p_{n'}$ has decided in $C\alpha\sigma_0\rho_\ell$ and $C\alpha\sigma_1\tau_\ell$, $\ell < n - n'-1$, and, for all $i \in \{0, \ldots, \ell\}$,
\begin{itemize}
\item $C\alpha\sigma_0\rho_i \widesim{p_{n'}} C\alpha\sigma_0\tau_i$,
\item if $i < \ell$, then $p_{n'}$ crashes exactly $i+1$ times in $\rho_i$ and $\tau_i$, and $p_{n'}$ crashes either $\ell$ or $\ell+1$ times in $\rho_\ell$ and $\tau_\ell$,
\item if $i < \ell$, then $\textit{value}(O, C\alpha\sigma_0\rho_i) = s_{v, n'+1+i+j}$ for some $v \in \{0, 1\}$ and some $j \geq 0$,
\item if $i < \ell$, then $\textit{value}(O, C\alpha\sigma_1\tau_i) = s_{v', n'+1+i+j}$ for some $v' \in \{0, 1\}$, and
\end{itemize}
Define $\rho_0 = \tau_0 = c_{n'}$.
Then $C\alpha\sigma_0\rho_0 \widesim{p_{n'}} C\alpha\sigma_1\tau_0$.
Recall that $\sigma_0$ and $\sigma_1$ are both permutations of $p_0, \ldots, p_{n'}$.
Therefore, if $\textit{value}(O, C\alpha) = s$, then either $\textit{value}(O, C\alpha\sigma_0) = s_{v, n'+1}$ and $\textit{value}(O, C\alpha\sigma_1) = s_{v', n'+1}$ (if $n'+1 < n$), or $\textit{value}(O, C\alpha\sigma_0) = \textit{value}(O, C\alpha\sigma_1) = s_\bot$ (if $n'+1 = n$).
Similarly, if $\textit{value}(O, C\alpha) = s_{v, j}$, then either $\textit{value}(O, C\alpha\sigma_0) = s_{v, j+n'+1}$ and $\textit{value}(O, C\alpha\sigma_1) = s_{v, j+n'+1}$ (if $j+n'+1 < n$), or $\textit{value}(O, C\alpha\sigma_0) = \textit{value}(O, C\alpha\sigma_1) = s_\bot$ (if $j+n'+1 \geq n$).
Finally, if $\textit{value}(O, C\alpha) = s_\bot$, then $\textit{value}(O, C\alpha\sigma_0) = \textit{value}(O, C\alpha\sigma_1) = s_\bot$.
In any of these cases, if $\textit{value}(O, C\alpha\sigma_0) = \textit{value}(O, C\alpha\sigma_1) = s_\bot$, then define $\ell = 0$.

Now suppose that we have constructed $\rho_i$ and $\tau_i$ where $0 \leq i < \ell$.
Then $\textit{value}(O, C\alpha\sigma_0\rho_i) = s_{v, n'+1+i+j}$ and $\textit{value}(O, C\alpha\sigma_1\tau_i) = s_{v', n'+1+i+j}$ for some $j \geq 0$.
Let $\delta$ be the $p_{n'}$-only execution from $C\alpha\sigma_0\rho_i$ in which $p_{n'}$ takes steps until either it decides a value or it is poised to apply an operation to $O$.
Since all of the objects except possibly $O$ have the same values in $C\alpha\sigma_0\rho_i$ and $C\alpha\sigma_1\tau_i$ and $C\alpha\sigma_0\rho_i \widesim{p_{n'}} C\alpha\sigma_1\tau_i$, there is a $p_{n'}$-only execution $\delta'$ from $C\alpha\sigma_1\tau_i$ such that $\delta \widesim{p_{n'}} \delta'$.
If $p_{n'}$ has decided in $C\alpha\sigma_0\rho_i\delta$ and $C\alpha\sigma_1\tau_i\delta'$, then define $\rho_{i+1} = \rho_i\delta$, $\tau_{i+1} = \tau_i\delta'$, and $\ell = i+1$.
Notice that $p_{n'}$ has crashed exactly $i+1 = \ell$ times in $\rho_{\ell}$ and $\tau_\ell$ in this case.
Otherwise, define $\rho_{i+1} = \rho_i\delta p_{n'}c_{n'}$ and $\tau_{i+1} = \rho_i\delta' p_{n'}c_{n'}$.
Notice that $C\alpha\sigma_0\rho_{i+1} \widesim{p_{n'}} C\alpha\sigma_1\tau_{i+1}$.
If $p_{n'}$ was poised to apply $op_R$ in $C\alpha\sigma_0\rho_i\delta$ and $C\alpha\sigma_1\tau_i\delta'$ or $n'+1+i+j = n-1$, then $\textit{value}(O, C\alpha\sigma_0\rho_i\delta p_{n'}c_{n'}) = \textit{value}(O, C\alpha\sigma_1\tau_i\delta'p_{n'}c_{n'}) = s_{\bot}$.
In this case, define $\ell = i+1$.
Otherwise, $p_{n'}$ is poised to apply $op_0$ or $op_1$ in these two configurations, $n'+1+i+j < n-1$, $\textit{value}(O, C\alpha\sigma_0\rho_i\delta) = s_{v, n'+1+i+j}$, and $\textit{value}(O, C\alpha\sigma_1\tau_i\delta') = s_{v', n'+1+i+j}$.
In this case, when $p_{n'}$ applies $op_0$ or $op_1$ in $C\alpha\sigma_0\rho_i\delta$, it changes the value of $O$ to $s_{v, n'+2+i+j}$.
Similarly, when $p_{n'}$ applies $op_0$ or $op_1$ in $C\alpha\sigma_1\tau_i\delta$, it changes the value of $O$ to $s_{v', n'+2+i+j}$.
Finally, $p_{n'}$ does not crash during $\delta$ or $\delta'$, so $p_{n'}$ crashes exactly $i+2$ times in $\rho_{i+1} = \rho_{i}\delta p_{n'}c_{n'}$ and $\tau_{i+1} = \tau_i\delta' p_{n'}c_{n'}$.
This completes the construction.

Define $\rho = \sigma_0\rho_\ell$ and $\tau = \sigma_1\tau_\ell$.
Since $p_{n'}$ crashes at most $\ell+1 < n-n'$ times in $\rho$ and $\tau$ and $p_0$ takes a step in $\sigma_0$ and $\sigma_1$, we know that the number of crash steps by $p_{n'}$ in any prefix of $\rho$ or $\tau$ is at most $n-n' \leq n-1$ times the number of steps collectively taken by $p_0, \ldots, p_{n'-1}$ in that prefix.
Since $n-1 \leq \lceil\frac{n}{2}\rceil\cdot 2 \leq \lceil\frac{n}{2}\rceil\cdot (n'+1) = m\cdot (n'+1)$, we know that $\rho, \tau \in \mathcal{E}^\star_m(C\alpha)$.
This completes the proof.
\end{proof}
\else
\fi

Lemma~\ref{lem:consn} and Lemma~\ref{lem:rconsn'} together say that the consensus number of $\mathcal{T}_{n, n'}$ is $n$ and the recoverable consensus number of $\mathcal{T}_{n, n'}$ is $n'$.
Hence, the ability of a non-readable type to solve recoverable wait-free consensus can be much lower than its ability to solve wait-free consensus.

\section{Conclusion}\label{sec:conclusion}

In this paper, we proved that the recoverable consensus hierarchy is robust for deterministic, readable types.
To accomplish this, we proved that Delporte-Gallet, Fatourou, Fauconnier, and Ruppert's \cite{dffr-22} $n$-recording condition is necessary for solving recoverable wait-free consensus among $n$ processes using deterministic types.
These authors originally proved that any deterministic, readable type with consensus number $n$ has recoverable consensus number at least $n-2$, though it remained an open question to determine whether there exists a type with consensus number $n$ and recoverable consensus number $n-2$.
However, for all $n \geq 4$, they defined a deterministic, readable type $\mathcal{X}_n$ that is $n$-discerning, $(n-2)$-recording, but not $(n-1)$-recording.
Our results show that $\mathcal{X}_n$ has recoverable consensus number $n-2$, resolving the open question above.

We also showed that the gap between the consensus number and recoverable consensus number of a deterministic, non-readable type can be arbitrarily large.
It is still unknown whether the recoverable consensus hierarchy is robust for all deterministic types.

\ifarxiv
\section*{Acknowledgements}
\else
\begin{acks}
\fi
I gratefully acknowledge the support of the Natural Sciences and Engineering Research Council of Canada (NSERC), PDF-578266-2023.
I also thank Trevor Brown for providing feedback on an earlier draft of this paper.
Finally, I thank the anonymous reviewers for their feedback.
\ifarxiv
\else
\end{acks}
\fi

\bibliographystyle{ACM-Reference-Format}
\bibliography{refs}

\end{document}